\documentclass[11pt]{article}

\usepackage[margin=1in]{geometry}
\usepackage{amsmath,amssymb,amsthm,amsxtra}
\usepackage{mathrsfs, mathtools}
\usepackage{bbm}
\usepackage{bm}
\usepackage[shortlabels]{enumitem}
\usepackage{microtype}
\microtypesetup{expansion=false}
\usepackage[T1]{fontenc}
\usepackage{xcolor}
\usepackage{comment}
\usepackage{float}
\usepackage{hyperref}
\usepackage{graphicx}
\usepackage{algorithm}
\usepackage[noend]{algpseudocode}
\makeatletter
\renewcommand{\theHALG@line}{\thealgorithm.\arabic{ALG@line}}
\makeatother

\hypersetup{hidelinks}
\numberwithin{equation}{section}

\newtheorem{theorem}{Theorem}[section]
\newtheorem{definition}[theorem]{Definition}
\newtheorem{lemma}[theorem]{Lemma}

\newtheorem{claim}[theorem]{Claim}

\newtheorem{proposition}[theorem]{Proposition}

\newcommand{\cG}{\mathcal{G}}

\newcommand{\cP}{\mathcal{P}}

\newcommand{\cZ}{\mathcal{Z}}

\newcommand{\wt}{\operatorname{w}}
\newcommand{\cmu}{\mu}

\newcommand{\Tr}{\text{Tr}}

\newcommand{\bbone}{\mathbbm{1}}

\newcommand{\I}{\mathcal{I}}

\makeatletter
\DeclareRobustCommand\Equiv{\mathrel{%
 \mathchoice
 {\Equiv@\textfont\displaystyle{.45}}
 {\Equiv@\textfont\textstyle{.45}}
 {\Equiv@\scriptfont\scriptstyle{.5}}
 {\Equiv@\scriptscriptfont\scriptscriptstyle{.55}}
}}
\newcommand{\Equiv@}[3]{%
 \rlap{\raisebox{#3\fontdimen5#12}{$\m@th#2 = $}}%
 \raisebox{-#3\fontdimen5#12}{$\m@th#2 = $}%
}
\makeatother

\newcommand{\TV}{\operatorname{\mathsf{TV}}}

\newcommand{\diag}{\operatorname{\mathsf{diag}}}

\newcommand{\benm}{\begin{enumerate}[leftmargin=*]}
\newcommand{\eenm}{\end{enumerate}}

\definecolor{forestgreen}{rgb}{0.13, 0.55, 0.13}

\newcommand{\op}{\operatorname{op}}

\newcommand{\one}{\mathbf 1}
\newcommand{\e}{\mathrm e}

\newcommand{\Id}{\mathrm{Id}}

\title{The Hard-Core Model on Bipartite Spectral Expanders: Counting and Sampling at All Fugacities}

\author{
Ijay Narang\thanks{School of Computer Science, Georgia Institute of Technology, \texttt{inarang3@gatech.edu}}
\and
Will Perkins\thanks{School of Computer Science, Georgia Institute of Technology, \texttt{wperkins3@gatech.edu}}
}
\date{}

\begin{document}
\maketitle

\begin{abstract}
We study approximate counting and sampling algorithms for the hard-core model on $\Delta$-regular bipartite graphs under a spectral expansion condition. Let $M_G$ be the biadjacency matrix of $G$. For every fixed $\xi\in(0,1)$, we give an FPRAS for the hard-core partition function and an efficient approximate sampler whenever
\[
\lambda\leq \frac{1-\xi}{\sigma_2(M_G)}.
\]

The main idea is to introduce a family of quadratic tilts in the left-right occupation imbalance and show that each tilted measure can be sampled efficiently using Glauber dynamics. A discrete Gaussian identity expresses the original hard-core model as an exact positive mixture of these tilted measures; truncation and simulated annealing then yield efficient counting and sampling algorithms.

For the complementary high-fugacity regime, we refine the polymer-model approach and show that the required phase-dominance and cluster expansion conditions follow from the singular-spectrum bound alone. Combining the two regimes, we obtain efficient approximate counting and sampling at every fugacity $\lambda>0$  whenever
\[
\sigma_2(M_G)\leq
c\left(\frac{\Delta^2}{\log(\mathrm e\Delta)}\right)^{1/3}
\]
for an absolute constant $c>0$. In particular, this recovers all-fugacity algorithms for random $\Delta$-regular bipartite graphs for all sufficiently large $\Delta$, while providing an efficiently verifiable certificate of their success on a given instance.
\end{abstract}

\newpage

\section{Introduction}
\label{sec:Intro}

Given a graph $G$, let $\I(G)$ denote the set of all independent sets of $G$. For a fugacity parameter $\lambda>0$, the hard-core model $\mu_{G,\lambda}$ is the probability distribution over $\I(G)$ defined by:
$$\mu_{G,\lambda}(I) = \frac{\lambda^{|I|}}{Z_G(\lambda)}, \qquad Z_G(\lambda):= \sum_{I \in \I(G)} \lambda^{|I|},$$
where $Z_G(\lambda)$ is the normalizing constant or partition function of the Gibbs measure $\mu_{G,\lambda}$.

Exact evaluation of \(Z_G(\lambda)\) is \(\#\mathrm{P}\)-hard, motivating the study of efficient approximate counting and sampling algorithms. Let
\(\varepsilon\in(0,1)\). An algorithm is an \emph{FPTAS} for
\(Z_G(\lambda)\) if it runs in time polynomial in \(|V(G)|\) and
\(1/\varepsilon\) and produces an estimate \(\widehat Z\) such that
\[
(1-\varepsilon)Z_G(\lambda)
\leq \widehat Z
\leq (1+\varepsilon)Z_G(\lambda) \,.
\]
A randomized algorithm is an \emph{FPRAS} for \(Z_G(\lambda)\) if it produces such an approximation with probability at least $3/4$.
An algorithm is an \emph{efficient approximate sampler} for
\(\mu_{G,\lambda}\) if it runs in time polynomial in \(|V(G)|\) and \(1/\varepsilon\) and outputs an independent set from a law
\(\widehat{\mu}\) satisfying
$\TV(\widehat{\mu},\mu_{G,\lambda})
<\varepsilon.$

Let $\mathcal{G}^{\le}_\Delta$ be the set of all graphs of maximum degree at most $\Delta$. For $G \in \mathcal{G}^{\le}_\Delta$, the computational tractability of approximating $Z_G(\lambda)$ exhibits a sharp transition. If $\lambda < \lambda_c(\Delta)$, then there is an FPTAS and efficient sampling scheme for the hard-core model~\cite{weitz2006counting}. If $\lambda > \lambda_c(\Delta)$, then there is no FPRAS or efficient sampling scheme unless $\mathbf{NP}=\mathbf{RP}$~\cite{Sly10,sly2012computational,galanis2016inapproximability}. The threshold $\lambda_c(\Delta) := \frac{(\Delta-1)^{\Delta-1}}{(\Delta-2)^\Delta}$ corresponds to the Gibbs uniqueness threshold for the hard-core model on the infinite $\Delta$-regular tree, and so in this case (and for all $2$-spin antiferromagnetic models), the computational transition coincides with a physics phase transition.

The complexity landscape is less understood when the input graph is required to be bipartite. Even at $\lambda =1$, approximately counting independent sets in bipartite graphs is the canonical problem $\#\textbf{BIS}$, which is not known either to admit an FPRAS or to be NP-hard to approximate.  Many other natural approximate counting problems  have been shown to be $\#\textbf{BIS}$-hard ~\cite{DyerGoldbergGreenhillJerrum2004, DyerGoldbergJerrum2010}, and so settling the complexity $\#\textbf{BIS}$ is one of the central problems in the field.

In \cite{cai2016hardness} it was shown that, given a bipartite graph $G \in \mathcal{G}^{\le}_\Delta$, approximating $Z_G(\lambda)$ is as hard as
$\#\textbf{BIS}$ whenever $\lambda>\lambda_c(\Delta)$. The
algorithmic side has seen progress for special classes of instances, such as lattices~\cite{helmuth2019algorithmic,cannon2026pirogov} and expander graphs~\cite{JenssenKeevashPerkins2020,liao2019counting,jenssen2023approximately,jenssen2026refined}. However, such results often require strong combinatorial hypotheses (such as small-set expansion) that can be difficult to verify. Thus, one would hope for a broader and efficiently checkable hypothesis.

A natural candidate for such a condition is spectral expansion. There are two recent reasons
for optimism. First, Chen, Chen, Chen, Yin, and
Zhang~\cite{CCCYZ2025} gave a general rapid-mixing criterion that, for
the hard-core model, can be checked using the least eigenvalue of the
adjacency matrix. For a $\Delta$-regular bipartite graph, however, the least adjacency
eigenvalue is always $-\Delta$, regardless of how well the graph
expands.  Our localization procedure can be viewed as a method for removing precisely this trivial bipartite direction before applying the criterion.

Second, recent work gives efficient sampling at every fugacity on random $\Delta$-regular bipartite graphs when $\Delta$ is a
sufficiently large constant. More precisely, Kocurek, Oveis Gharan, and
Tjowasi \cite{kocurek2026sampling} handle the regime $\lambda = O \left(  \frac{1}{\sqrt{\Delta}} \right)$ by analyzing two auxiliary Markov chains, whereas the polymer-model framework of Jenssen, Keevash, and Perkins \cite{JenssenKeevashPerkins2020} and refinements in~\cite{chen2022sampling} handles the regime
$\lambda =  \Omega \left ( \frac{\log \Delta}{\Delta} \right)$.  These proofs, however, use properties of random regular bipartite graphs beyond a bound on the
nontrivial spectrum. Nevertheless, this suggests a natural question: can all-fugacity algorithms be obtained from an efficiently checkable condition on the nontrivial spectrum alone? Our results answer this question affirmatively.

For notation, let $\mathcal{G}_{\Delta, \Delta}$ denote the set of all $\Delta$-regular bipartite graphs on vertex set $L\sqcup R$, where $|L|=|R|=n$. Given a graph $G$, let $M_G$ denote its biadjacency matrix and $\sigma_2(M_G)$ denote its second-largest singular value. Define the class of graphs $\mathcal{G}_{\Delta, \Delta}^\sigma$ to be

\begin{equation}
\mathcal{G}_{\Delta, \Delta}^\sigma := \{G \in \mathcal{G}_{\Delta, \Delta} : \sigma_2(M_G) \le \sigma\}.
\end{equation}

Our main result is the following.

\begin{theorem}
\label{thm:all_fug}
There are absolute constants \(c>0\) and
\(\Delta_0\geq3\) such that, for every
\(\Delta\geq\Delta_0\) and every
\(G\in\mathcal{G}_{\Delta,\Delta}^{\sigma}\), where
\begin{equation}
\sigma
\leq
c
\left(
\frac{\Delta^2}{\log(\e\Delta)}
\right)^{1/3},
\label{eq:all-fugacity-spectral-range}
\end{equation}
there is an FPRAS for \(Z_G(\lambda)\) and an efficient approximate
sampler for \(\mu_{G,\lambda}\) at every fugacity \(\lambda>0\).
\end{theorem}

An immediate consequence of Theorem \ref{thm:all_fug} is a
certifying improvement for random regular bipartite graphs. Indeed, if
\(G\) is a uniformly random \(\Delta\)-regular bipartite graph, then
\[
\sigma_2(M_G)\leq 2\sqrt{\Delta-1}+o_n(1)
\]
with probability \(1-o_n(1)\) for \(\Delta\) sufficiently large \cite{BritoDumitriuHarris2022,friedman2003proof},
this satisfies the condition of the all-fugacity Theorem~\ref{thm:all_fug}.  

Previous all-fugacity results for random regular bipartite graphs were obtained by combining algorithms whose guarantees held with high probability over the random instance~\cite{kocurek2026sampling,JenssenKeevashPerkins2020,chen2022sampling}, but without an efficient way to verify these guarantees. Our theorem instead provides a single efficiently checkable spectral condition that certifies the complete all-fugacity guarantee for the given graph.  See, e.g., Remark 1 of~\cite{JenssenKeevashPerkins2020} or Remark 2.1 of~\cite{chen2022sampling} which ask for algorithms with the stronger type of guarantee we provide here. 

The proof of Theorem \ref{thm:all_fug} follows from separate analyses of the moderate and high fugacity regimes. In the moderate fugacity regime, we utilize a new technique, which we call discrete quadratic localization, to prove the following.

\begin{theorem}
\label{thm:main}
Fix an integer \(\Delta\geq3\), \(\xi\in(0,1)\), and \(\sigma>0\).
Suppose that \(\lambda>0\) satisfies
\begin{equation}
\lambda\leq\frac{1-\xi}{\sigma}.
\end{equation}
Then there is an FPRAS for approximating \(Z_G(\lambda)\) and an
efficient approximate sampler for \(\mu_{G,\lambda}\) for every
\(G\in\mathcal{G}_{\Delta,\Delta}^{\sigma}\).
\end{theorem}
For a near-Ramanujan bipartite graph, $\sigma= O (\sqrt{\Delta} )$, and so this theorem applies up to fugacity $O ( 1/\sqrt{\Delta})$, matching the bound obtained by recent algorithms for random regular bipartite graphs~\cite{kocurek2026sampling}, but under a deterministic and efficiently checkable condition.

In the high fugacity regime, we refine the cluster expansion analysis of \cite{JenssenKeevashPerkins2020} to follow from spectral expansion. 

\begin{theorem}
\label{thm:cluster-expansion}
There exists an absolute constant \(C>0\) such that the following holds. Fix an integer \(\Delta\geq 3\) and
\(0\leq\sigma<\Delta\), and suppose that \(\lambda>0\) satisfies
\begin{equation}
\lambda
\geq
\exp\!\left[
C\frac{\sigma^2}{\Delta^2-\sigma^2}\log(\e\Delta)
\right]-1.
\label{eq:cluster-expansion-lambda}
\end{equation}
Then there is an FPTAS for approximating \(Z_G(\lambda)\) and an
efficient approximate sampler for \(\mu_{G,\lambda}\) for every
\(G\in\mathcal{G}_{\Delta,\Delta}^{\sigma}\).
\end{theorem}
When $\sigma = O(\sqrt{\Delta})$, the lower bound is $\Omega \left( \frac{\log \Delta}{\Delta} \right)$, matching the results of the strongest random-regular polymer model algorithms~\cite{chen2022sampling}.

\subsection{Discrete Quadratic Localization and the Moderate-Fugacity Regime}

When $\lambda>\lambda_c(\Delta)$, the single-site Glauber dynamics for the hard-core measure need
not mix rapidly on a bipartite graph in $\cG_{\Delta,\Delta}$. This is witnessed by the random $\Delta$-regular bipartite graph:
the hard-core measure places significant mass on both a left-heavy region and a right-heavy region and moving between these regions requires passing through balanced configurations with exponentially small weight~\cite{mossel2009hardness}.

This obstruction presents itself spectrally as well. Because $G$ is $\Delta$-regular, the biadjacency matrix admits the decomposition $M=\frac{\Delta}{n}J+M_{\perp}$ where $\|M_{\perp}\|_{\op}\leq \sigma.$ Writing out the pinned influence matrix $\Psi_S$ from the spectral mixing criterion of \cite{CCCYZ2025}, this obstruction manifests in the exceptional rank-one term $(\Delta/n)J$. To circumvent this issue, we introduce a negative quadratic field which localizes
the left-right imbalance to a particular window so that its discrete curvature exactly cancels the rank-one contribution $(\Delta/n)J$. The remaining dependency is then governed by $\sigma$.

More concretely, define $
m(I):=|I\cap L|-|I\cap R|,$ $
q:=1-\frac{\Delta}{n},$ and $
f(a):=\frac{a^2}{2}$.
For $t>0$ and $k\in\mathbb Z$, the quadratically localized measure is given by
\[
\mu_{G,t,k}(I)
= \frac{t^{|I|}q^{f(m(I))+km(I)}}{\cZ_{G,k}(t)}, \qquad \cZ_{G,k}(t) = \sum_{I \in \I(G)} t^{|I|}q^{f(m(I))+km(I)}.
\]

When $k$ is very negative, the measure favors left-heavy independent
sets, when $k=0$, it favors balanced independent sets, and when $k$ is very positive, it favors right-heavy independent sets. Being attracted to a given balance prevents Glauber dynamics from crossing between the left-heavy and right-heavy distributions, giving rapid mixing.

We now discuss how to reconstruct the hard-core distribution from the above measure. Define $p_k := \frac{q^{k^2/2}}{\displaystyle\sum_{r\in\mathbb{Z}} q^{r^2/2}}$ so that $\{p_k\}_{k \in \mathbb{Z}}$ defines a probability distribution on the integers and denote this measure by $p$. The key identity connecting $\mu_{G,t,k}$ to the hard-core model is as follows:
\begin{align} \label{eq:inf_sum_intro}
Z_G(t)
&= \sum_{I\in\I(G)} t^{|I|}
= \sum_{I\in\I(G)} t^{|I|}
\mathbb E_{k\sim p}
\left[q^{f(m(I))+km(I)}\right] \nonumber \\
&= \mathbb E_{k\sim p}
\left[
\sum_{I\in\I(G)}
t^{|I|}q^{f(m(I))+km(I)}
\right]
= \mathbb E_{k\sim p}\!\left[\cZ_{G,k}(t)\right]
= \sum_{k\in\mathbb Z} p_k \cZ_{G,k}(t).
\end{align}

Equation \eqref{eq:inf_sum_intro} can be viewed as a discrete version of the Hubbard--Stratonovich transform \cite{Hub59}, which has been deeply studied in the context of Ising models \cite{BL02,BB19,EKZ21,AnariEtAl21,KLR22, AKV2024}.

The work most closely related to our setting is that of \cite{KLR22}, which studies sampling from Ising models whose Hamiltonians have a small number of unfavorable eigen-directions. Their approach identifies these eigen-directions and integrates them out via the Hubbard--Stratonovich transform. The key distinction is that, in our setting, the relevant obstructions are not encoded directly in the Hamiltonian, but instead emerge in the dependency matrix that we seek to control.

We also note that the idea of utilizing quadratic localization to improve sampling efficiency appears in the literature on umbrella sampling \cite{TV77,Kas11} and proximal sampling \cite{LST21}.

 Our proof of Theorem \ref{thm:main} is a direct instantiation of discrete quadratic localization. To show that  the method can be implemented in polynomial time, we establish spectral independence for $\mu_{G,t,k}$ by a theorem from \cite{CCCYZ2025}, giving rapid mixing of Glauber dynamics. To convert the Hubbard--Stratonovich transform into an efficient algorithm, we show that truncating the infinite sum appropriately gives a sufficiently good approximation. The algorithm then follows from applying standard annealing techniques to obtain an  FPRAS and efficient sampler. 

 \subsection{Spectrally Certified Polymer Models at High Fugacity}
 
 For large fugacity, we use the two-phase polymer framework of Jenssen, Keevash, and Perkins~\cite{JenssenKeevashPerkins2020}. We define two abstract polymer models intended to capture the left- and right-dominant independent sets respectively.  We then show two things: (i) the scaled sum of polymer model partition functions is a good approximation of the hard-core partition function (ii) these polymer models have convergent cluster expansions. The polymer models we use are essentially the same as those used in~\cite{JenssenKeevashPerkins2020,liao2019counting,chen2022sampling}. Our main contribution (and the reason we can obtain improved guarantees) is showing that the approximation by polymer models and polymer-convergence criterion follow from a bound on the nontrivial singular spectrum alone (while \cite{JenssenKeevashPerkins2020,liao2019counting} used a small-set expansion condition and~\cite{chen2022sampling} used a second-moment analysis of the random graph).

We first prove that the hard-core measure is concentrated near two phases: independent sets that are predominantly supported on \(L\), and
independent sets that are predominantly supported on \(R\). To establish this, we apply a long-range order theorem of Hadas and Peled~\cite{HP2026}, which shows that phase dominance follows from local and global expansion certificates. When $\sigma$ is far from $\Delta$, we show that these certificates are met by bounding short random-walk return probabilities. When \(\sigma\) is close to \(\Delta\), we combine a simple one-edge local certificate with the Cheeger bound.

To show convergence of the cluster expansion for the polymer models, we verify the Koteck\'y--Preiss condition by utilizing a scale-sensitive Tanner
inequality, lower bounding the neighborhood size of every polymer in
terms of \(\Delta\) and \(\sigma\):
\begin{equation}
    \label{eqTanner}
   \frac{|N(S)|}{|S|}
\geq
\frac{\Delta^2}
{\sigma^2+(\Delta^2-\sigma^2)|S|/n} \,, 
\end{equation}
where $S$ is a subset of one bipartition class of $G$ and $G$ has second singular value at most $\sigma$.
This forces the polymer weights to
decay sufficiently rapidly for the Koteck\'y--Preiss condition to hold.

\subsection{Organization of the Paper} 
Section \ref{sec:prelim} presents preliminaries on spectral graph theory, the Hadas--Peled long-range order criterion, Markov chains, and polymer models. The proofs and algorithms for Theorems \ref{thm:main} and \ref{thm:cluster-expansion} are provided in Sections \ref{sec:main} and \ref{sec:cluster-proof} respectively. We combine Theorems \ref{thm:main} and \ref{thm:cluster-expansion} to attain Theorem \ref{thm:all_fug} in Subsection \ref{subsec:param_matching}.

\section{Preliminaries}
\label{sec:prelim}

\subsection{Spectral graph theory}

Throughout, fix $\sigma\geq 0$ and let \(G=(L\sqcup R,E)\in\mathcal{G}_{\Delta,\Delta}^{\sigma}\),
where \(|L|=|R|=n\), and \(M:=M_G\in\{0,1\}^{n\times n}\) denotes its biadjacency matrix. We write \(\Id\) for the identity matrix, \(J\) for the all-ones matrix,
and \(\one\) for the all-ones vector. For a real matrix \(A\), its
Euclidean operator norm and Frobenius norm are
\[
\|A\|_{\mathrm{op}}
=
\sup_{x\ne0}\frac{\|Ax\|_2}{\|x\|_2},
\qquad
\|A\|_F
=
\bigl(\Tr(A^{\mathsf T}A)\bigr)^{1/2}.
\]

Since \(G\) is \(\Delta\)-regular,
$
M\one=M^{\mathsf T}\one=\Delta\one.
$
Thus the largest singular value of \(M\) is \(\Delta\). Let $\Pi_\perp=\Id-\frac1nJ$ be the orthogonal projection away from the constant direction. The second singular value satisfies
\[
\sigma_2(M)=\|\Pi_\perp M\Pi_\perp\|_{\mathrm{op}}\leq\sigma
\]
by the definition of $\mathcal G_{\Delta,\Delta}^{\sigma}$. A simple observation is the following.

\begin{lemma}
\label{lem:trace-comparison}
If $n\geq\Delta^2$ and $\sigma<\Delta$, then
\begin{equation}
\frac{\sigma^2}{\Delta^2-\sigma^2}
\geq\frac1\Delta.
\label{eq:theta-lower}
\end{equation}
\end{lemma}

\begin{proof}
Because $M$ has exactly $n\Delta$ entries equal to one, $\|M\|_F^2=n\Delta.$ If the singular values are
$\Delta=\sigma_1(M)\geq\sigma_2(M)\geq\cdots$, then
\[
n\Delta
=
\Delta^2+\sum_{i=2}^n\sigma_i(M)^2
\leq
\Delta^2+(n-1)\sigma_2(M)^2.
\]
giving that $\sigma_2(M)^2 \geq \frac{\Delta(n-\Delta)}{n-1}.$ Since $\sigma_2(M)\leq\sigma$, the same lower bound holds for $\sigma^2$.
The right-hand side is increasing in $n$, and at $n=\Delta^2$ it equals
$\Delta^2/(\Delta+1)$. Hence $\sigma^2\geq\frac{\Delta^2}{\Delta+1}$ and substitution into the increasing function
$x\mapsto x/(\Delta^2-x)$ gives the result.
\end{proof}

One reason spectral expansion is useful is that it implies notions of combinatorial
expansion (see for example, \cite{Chung1997}). For a $\Delta$-regular graph $G=(V,E)$, define its edge-Cheeger constant by
\begin{equation}
h(G)
:=
\min_{\substack{\varnothing\neq S\subseteq V\\|S|\leq|V|/2}}
\frac{|E(S,V\setminus S)|}{|S|}.
\label{eq:def-cheeger}
\end{equation}

\begin{lemma}
\label{lem:cheeger}
Let $G \in \mathcal{G}_{\Delta, \Delta}^\sigma$. The edge-Cheeger constant of $G$ satisfies
\begin{equation}
h(G)
\geq
\frac{\Delta-\sigma}{2}.
\label{eq:cheeger-spectral}
\end{equation}
\end{lemma}

\begin{proof}
If $G$ is disconnected, then $\sigma_2(M)=\Delta\leq\sigma$ and the claimed bound is
trivial. Suppose that $G$ is connected. The eigenvalues of the adjacency matrix
\[
\begin{pmatrix}
0&M\\ M^{\mathsf T}&0
\end{pmatrix}
\]
are the positive and negative singular values of $M$. The largest
eigenvalue is $\Delta$ and the second-largest is
$\sigma_2(M)\leq\sigma$. The standard Cheeger inequality for a $\Delta$-regular graph gives $h(G)\geq(\Delta-\sigma_2(M))/2\geq(\Delta-\sigma)/2$
\cite{HLW2006}.
\end{proof}

In addition, spectral expansion provides finer control of vertex expansion at a specified set size. The following scale-sensitive form of Tanner's inequality provides this control.

The following is the bipartite singular-value form of Tanner’s scale-sensitive expansion inequality \cite{Tanner1984}, see also \cite[Theorem 4.15]{HLW2006}.
\begin{lemma}
\label{lem:tanner}
Let $G \in \mathcal{G}_{\Delta, \Delta}^\sigma$. For every nonempty $S\subseteq L$,
\begin{equation}
\frac{|N(S)|}{|S|}
\geq
\frac{\Delta^2}
{\sigma^2+(\Delta^2-\sigma^2)|S|/n}.
\label{eq:size-sensitive-tanner}
\end{equation}
The same inequality holds for nonempty subsets of $R$.
\end{lemma}

\begin{proof}
Let $x:=|S|/n$ and decompose
\[
\one_S=x\one+u,
\qquad
\langle u,\one\rangle=0.
\]
Since $M^{\mathsf T}\one=\Delta\one$, and since the operator norm of
$M^{\mathsf T}$ on $\one^\perp$ is at most $\sigma$, the two summands in
\[
M^{\mathsf T}\one_S=\Delta x\one+M^{\mathsf T}u
\]
are orthogonal. Hence
\begin{align}
\|M^{\mathsf T}\one_S\|_2^2
&=
\Delta^2x^2n+\|M^{\mathsf T}u\|_2^2
\nonumber\\
&\leq
\Delta^2x^2n+\sigma^2\|u\|_2^2
\nonumber\\
&=
n\bigl(\Delta^2x^2+\sigma^2x(1-x)\bigr).
\label{eq:tanner-upper}
\end{align}
The vector $M^{\mathsf T}\one_S$ is supported on $N(S)$ and has
coordinate sum $\Delta|S|$. Hence Cauchy--Schwarz gives
\begin{equation}
\|M^{\mathsf T}\one_S\|_2^2
\geq
\frac{\Delta^2|S|^2}{|N(S)|}.
\label{eq:tanner-lower}
\end{equation}
Combining \eqref{eq:tanner-upper} and \eqref{eq:tanner-lower}, then
substituting $x=|S|/n$, proves \eqref{eq:size-sensitive-tanner}.
The argument for subsets of $R$ uses $M$ in place of
$M^{\mathsf T}$.
\end{proof}

\subsection{Long-range order}

We recall the following results of Hadas and Peled \cite{HP2026}. In
particular, they define a local expansion property through a random
connected subgraph.

\begin{definition}[Local expansion; \cite{HP2026}]
\label{def:local-expansion}
Let $H=(V,E)$ be a finite simple $\Delta$-regular graph. We say that
$H$ satisfies the local expansion property with constants 
$C_{\mathrm{LE}},M_{\mathrm{LE}}>0$ if there is a probability
distribution on connected subgraphs $\boldsymbol T$ of $H$, each
having at least one edge, such that for every edge $e\in E$ and every
vertex $v\in V$,
\begin{align}
|E|\,\mathbb P(e\in E(\boldsymbol T))
&\leq M_{\mathrm{LE}},
\label{eq:LE-edge}\\
|V|\,\mathbb P\bigl(N(v)\cap V(\boldsymbol T)\neq\varnothing\bigr)
&\geq \frac{\Delta M_{\mathrm{LE}}}{C_{\mathrm{LE}}}.
\label{eq:LE-vertex}
\end{align}
\end{definition}

Note that every finite, simple, $\Delta$-regular graph satisfies the local expansion property with constants $C_{\mathrm{LE}}=M_{\mathrm{LE}}=1$. Indeed, let $\boldsymbol T$ consist of one edge chosen uniformly at
random. Inequality~\eqref{eq:LE-edge}  then holds with equality.
Moreover, the $\Delta^2$ edge incidences at the vertices of $N(v)$
involve at least $\Delta^2/2$ distinct edges, whereas
$|E|=\Delta|V|/2$, hence
\[
|V|\,\mathbb P\bigl(N(v)\cap V(\boldsymbol T)\neq\varnothing\bigr)
\geq \Delta.
\]

When $H$ is bipartite, a random walk gives an alternative certificate.
\begin{proposition}
Let $H$ be a finite, simple $\Delta$-regular bipartite graph, and let $P$ be the transition matrix of the simple random walk on a bipartite graph $H$.  Then $H$ satisfies the local expansion property of Definition~\ref{def:local-expansion}  with constants $M_{\mathrm{LE}}=\Delta$  and
\begin{equation}
C_{\mathrm{LE}} = 1 + \Delta \cdot \max_{v\in V(H)} \sum_{j=1}^{\lfloor(\Delta-1)/2\rfloor}P^{2j}(v,v).
\label{eq:HP-random-walk-certificate}
\end{equation}
\end{proposition}
\begin{proof}
  Apply \cite[Lemma 1.6]{HP2026} with $M_0 = \Delta$. Since $H$ is bipartite, odd return probabilities are $0$, which gives
\eqref{eq:HP-random-walk-certificate}.    
\end{proof}

These certificates allow us to apply the Hadas--Peled criterion to bound the total weight of independent sets having substantial intersection size with both bipartition classes. In particular, the second part of Theorem~1.7 of
Hadas--Peled~\cite{HP2026}, written in the present notation, gives the following.
\begin{theorem}[Theorem 1.7 of \cite{HP2026}]
\label{prop:HP}
There are absolute constants $C_{\mathrm{HP}},c_{\mathrm{HP}}>0$ such
that the following holds. Let $\Delta\geq2$, let $\lambda>0$, and let
$H=(L_H\sqcup R_H,E_H)$ be a finite simple $\Delta$-regular bipartite
graph. Suppose that $H$ has edge-Cheeger constant
$h(H)>0$ and local expansion parameters
$C_{\mathrm{LE}},M_{\mathrm{LE}}>0$.  Let
$1\leq r\leq |L_H|/2$. If
\begin{equation}
\log(1+\lambda)
>
C_{\mathrm{HP}}\frac{C_{\mathrm{LE}}}{\Delta}
\max\left\{
\log\!\left(
C_{\mathrm{HP}}\frac{\Delta}{h(H)}
\frac{|V(H)|}{r}
\right),
\frac{\Delta}{h(H)}
\frac{|V(H)|}{rM_{\mathrm{LE}}}
\right\},
\label{eq:HP-hypothesis}
\end{equation}
then
\begin{equation}
\sum_{\substack{I\in\I(H)\\
|I\cap L_H|>r,\ |I\cap R_H|>r}}
\lambda^{|I|}
\leq
(1+\lambda)^{|L_H|}
\exp\!\left[
-c_{\mathrm{HP}}\frac{h(H)}{\Delta}\,
r\log(1+\lambda)
\right].
\label{eq:HP-conclusion}
\end{equation}
\end{theorem}

\subsection{Glauber dynamics and a dependency-matrix criterion}

Let \(P\) be the transition matrix of a finite reversible Markov chain
with stationary distribution \(\pi\). Its mixing time is
\[
T_{\mathrm{mix}}(\varepsilon)
=
\min\left\{
t\ge0:
\max_x\TV\bigl(P^t(x,\cdot),\pi\bigr)\le\varepsilon
\right\} \,,
\]
where $
\TV(\mu,\nu)
=
\frac12\sum_x|\mu(x)-\nu(x)|$ is the total variation distance.
We write \(T_{\mathrm{mix}}:=T_{\mathrm{mix}}(1/4)\).

One of the most widely studied Markov chains for sampling from
high-dimensional probability distributions is the single-site Glauber
dynamics. In this paper, we use its specialization to distributions on downward-closed subsets of 
the Boolean hypercube. Let \(m\in\mathbb N\), let \(\mathcal F\subseteq2^{[m]}\) be a nonempty downward-closed
family, and let \(\mu\) be a full-support probability distribution on
\(\mathcal F\). We identify each
configuration in \(\{0,1\}^m\) with the corresponding subset of
\([m]\). From the current state \(S\in\mathcal F\), the Glauber dynamics chooses a
coordinate \(v\in[m]\) uniformly at random and sets
\[
T:=S\setminus\{v\}.
\]
It then resamples whether \(v\) is present according to its conditional
distribution under \(\mu\), given the configuration \(T\) on all other
coordinates. More precisely, the next state is
\[
\begin{cases}
T\cup\{v\},
  &\displaystyle\text{with probability }
  \frac{\mu(T\cup\{v\})}
       {\mu(T)+\mu(T\cup\{v\})},\\[1.2ex]
T,
  &\displaystyle\text{with probability }
  \frac{\mu(T)}
       {\mu(T)+\mu(T\cup\{v\})}.
\end{cases}
\]
Here we interpret \(\mu(A)=0\) whenever \(A\notin\mathcal F\).

For \(S\in\mathcal F\), let
\[
U_S
=
\{v\in[m]\setminus S:S\cup\{v\}\in\mathcal F\}
\]
be the elements that remain available after conditioning on \(S\). The conditional distribution on these remaining elements is called a \emph{pinning}. For every nonmaximal \(S\), define the marginal ratios
\[
r_S(v)
=
\frac{\mu(S\cup\{v\})}{\mu(S)},
\qquad v\in U_S,
\]
and the pairwise dependency matrix \(\Psi_S\in\mathbb R^{U_S\times U_S}\) by
\[
\Psi_S(u,v)
=
\begin{cases}
\displaystyle
\frac{\mu(S\cup\{u,v\})\mu(S)}
 {\mu(S\cup\{u\})\mu(S\cup\{v\})}-1,
&u\ne v,\\[2ex]
0,&u=v.
\end{cases}
\]
Here \(\mu(T)\) is interpreted as zero when \(T\notin\mathcal F\).
Finally, write
\[
r_{\max}
=
\max_{\substack{S\in\mathcal F\\v\in U_S}}r_S(v),
\qquad
\mu_{\min}
=
\min_{S\in\mathcal F}\mu(S).
\]

\begin{theorem}[{\cite[Theorem~1.9]{CCCYZ2025}}]
\label{thm:chen-general}
Let \(\mathcal F\subseteq2^{[m]}\) be a nonempty downward-closed
family, and let \(\mu\) be a full-support probability distribution on
\(\mathcal F\). Suppose that, for some \(\delta\in(0,1)\),
\[
\Psi_S
\preceq
\Id+(1-\delta)\diag(r_S)^{-1}
\]
for every nonmaximal \(S\in\mathcal F\). Then the single-site Glauber
dynamics for \(\mu\) satisfies
\[
T_{\mathrm{mix}}
=
O\!\left(
(1+r_{\max})\frac{m}{\delta}
\log\frac1{\mu_{\min}}
\right).
\]
\end{theorem}
See also~\cite{GJMMPS26} for a short, alternative proof of Theorem~\ref{thm:chen-general} adapting the Bochner--Bakry--\'{E}mery approach of~\cite{kondratiev2013spectral} to the discrete setting.

\subsection{Polymer models}

We briefly recall the polymer models used in our high-fugacity argument. For background on algorithmic applications of polymer models and cluster expansions, see
\cite{helmuth2019algorithmic,JenssenKeevashPerkins2020}.

Let \(H\) be a finite graph, called the \emph{host graph}. A polymer
model on \(H\) is specified by a collection \(\cP=\cP(H)\) of nonempty
connected subsets of \(V(H)\), together with a weight
\(w_\gamma\in\mathbb C\) for each \(\gamma\in\cP\). Two polymers
\(\gamma,\gamma'\in\cP\) are said to be compatible if $\operatorname{dist}_H(\gamma,\gamma')>1,$
and incompatible otherwise. Let \(\Omega(\cP)\) denote the collection
of all sets of pairwise compatible polymers, including the empty set.
The polymer partition function is
\begin{equation}
\Xi
:=
\sum_{\Gamma\in\Omega(\cP)}
\prod_{\gamma\in\Gamma}w_\gamma .
\label{eq:polymer-partition-function}
\end{equation}
When all polymer weights are nonnegative, the associated polymer Gibbs
distribution is
\[
\nu(\Gamma)
=
\frac{1}{\Xi}
\prod_{\gamma\in\Gamma}w_\gamma,
\qquad
\Gamma\in\Omega(\cP).
\]

To describe the cluster expansion, consider an ordered tuple
\(\boldsymbol{\gamma}=(\gamma_1,\ldots,\gamma_k)\in\cP^k\), where
repetitions are allowed. Its incompatibility graph
\(F_{\boldsymbol{\gamma}}\) is the graph on vertex set \([k]\) in which
distinct \(i,j\in[k]\) are adjacent exactly when
\(\gamma_i\) and \(\gamma_j\) are incompatible. Define
\[
\phi(\boldsymbol{\gamma})
:=
\frac{1}{k!}
\sum_{\substack{
A\subseteq E(F_{\boldsymbol{\gamma}})\\
([k],A)\ \mathrm{connected}
}}
(-1)^{|A|},
\]
with the convention that the graph on one vertex is connected. Thus
\(\phi(\boldsymbol{\gamma})=0\) unless
\(F_{\boldsymbol{\gamma}}\) is connected.  The cluster expansion is the formal power series 
\begin{equation}
\log\Xi
=
\sum_{k\geq 1}
\sum_{\boldsymbol{\gamma}\in\cP^k}
\phi(\boldsymbol{\gamma})
\prod_{i=1}^k w_{\gamma_i}.
\label{eq:cluster-expansion}
\end{equation}
A tuple with connected incompatibility graph is called a
\emph{cluster}. We write $|\boldsymbol{\gamma}|
:=
\sum_{i=1}^k|\gamma_i|$ for its total size. We use the following form of the Koteck\'y--Preiss convergence criterion; see \cite{KoteckyPreiss1986} and
\cite[Theorem~6]{JenssenKeevashPerkins2020}.

\begin{theorem}[Koteck\'y--Preiss criterion]
\label{thm:KP}
Let \((\cP,w)\) be a graph-based polymer model on \(H\), and let
\(\mathfrak g:\cP\to[0,\infty)\). Suppose that every
\(\gamma\in\cP\) satisfies
\begin{equation}
\sum_{\substack{\gamma'\in\cP\\
\operatorname{dist}_H(\gamma,\gamma')\leq 1}}
|w_{\gamma'}|
\exp\!\left(
|\gamma'|+\mathfrak g(\gamma')
\right)
\leq |\gamma|.
\label{eq:KP-criterion}
\end{equation}
Then the cluster expansion \eqref{eq:cluster-expansion} converges
absolutely. Moreover, extending \(\mathfrak g\) additively to clusters
by $\mathfrak g(\boldsymbol{\gamma})
:=
\sum_{i=1}^k\mathfrak g(\gamma_i),$
we have, for every \(v\in V(H)\),
\begin{equation}
\sum_{\substack{
k\geq1,\ \boldsymbol{\gamma}\in\cP^k\\
F_{\boldsymbol{\gamma}}\ \mathrm{connected}\\
v\in\bigcup_{i=1}^k\gamma_i
}}
|\phi(\boldsymbol{\gamma})|
\prod_{i=1}^k|w_{\gamma_i}|
\exp\!\left(
\mathfrak g(\boldsymbol{\gamma})
\right)
\leq 1.
\label{eq:KP-tail-bound}
\end{equation}
\end{theorem}
We call $\mathfrak g$ the \emph{decay function} of the polymer model as it controls the convergence rate of the cluster expansion.
In particular, suppose that
\(\mathfrak g(\gamma)\geq\rho|\gamma|\) for every polymer
\(\gamma\). Summing \eqref{eq:KP-tail-bound} over \(v\in V(H)\)  shows that 
\begin{equation}
\sum_{\substack{
k\geq1,\ \boldsymbol{\gamma}\in\cP^k\\
F_{\boldsymbol{\gamma}}\ \mathrm{connected}\\ |\boldsymbol{\gamma} | \ge m
}}
|\phi(\boldsymbol{\gamma})|
\prod_{i=1}^k|w_{\gamma_i}|
\leq
|V(H)|e^{-\rho m}.
\label{eq:cluster-tail}
\end{equation}
Thus the cluster expansion may be truncated at logarithmic size while
incurring only an exponentially small error. The following algorithmic consequence packages the counting and
sampling results of
\cite[Theorems~8 and~9]{JenssenKeevashPerkins2020}, which build on the
algorithmic cluster-expansion framework of
\cite{helmuth2019algorithmic}.

\begin{proposition}[Algorithmic polymer criterion]
\label{prop:JKP}
Fix \(D\in\mathbb N\). Consider a family of graph-based polymer models
on host graphs \(H\) of maximum degree at most \(D\) with decay functions $\mathfrak g$, and suppose that
all polymer weights are nonnegative. Assume that the following
properties hold uniformly over the family:
\begin{enumerate}
    \item There are constants \(c_1,c_2>0\) such that, given \(H\) and
    a connected set \(\gamma\subseteq V(H)\), one can determine whether
    \(\gamma\in\cP(H)\) and, when it is, compute
    \(w_\gamma\) and \(\mathfrak g(\gamma)\) in time
    \[
    O\!\left(
    |\gamma|^{c_1}e^{c_2|\gamma|}
    \right).
    \]

    \item There is a constant \(\rho>0\), independent of \(H\), such
    that
    \[
    \mathfrak g(\gamma)\geq\rho|\gamma|
    \qquad
    \text{for every }\gamma\in\cP(H).
    \]

    \item The Koteck\'y--Preiss condition
    \eqref{eq:KP-criterion} holds with decay function $\mathfrak g$.
\end{enumerate}
Then the polymer partition function \(\Xi\) admits an FPTAS, and the
polymer Gibbs distribution admits an efficient sampler.
\end{proposition}

\section{Proof of Theorem \ref{thm:main}}
\label{sec:main}

The key idea of the algorithm is to sample from a weighted distribution on independent sets of $G$ that introduces a quadratic penalty on the left-right occupation
imbalance. The discrete Hessian of this penalty cancels the rank-one part \((\Delta/n)J\) of \(M\), leaving an all-pinnings dependency bound
controlled only by \(\lambda\sigma\). We are then able to reconstruct the original hard-core weights coefficient-wise through a truncated infinite positive mixture, obtained from a discrete Gaussian identity.

Throughout this section, let $G\in\mathcal G_{\Delta,\Delta}^{\sigma}$. We assume that $n \ge \Delta/\xi$, as if not then we can simply enumerate over all independent sets. We begin by setting
\begin{equation}\label{eq:q-def}
  q:=1-\frac{\Delta}{n}\in[1-\xi,1).
\end{equation}
Additionally, for an independent set \(I\), write $m(I):=|I\cap L|-|I\cap R|.$

\subsection{Quadratically localized measures}

We begin by defining the weighted measure we will sample from. For notation, we write $f(a):=\frac{a^2}2$. Given a graph $G$, \(t>0\), and an integer \(k\), we define a weight function $\wt_{G,t,k}$ on independent sets as
\begin{align}
  \wt_{G,t,k}(I)
  &:=t^{|I|}q^{f(m(I))+km(I)},\label{eq:component-weight}
\end{align}
We use this weight function to define a partition function parameterized by $t$ and $k$ and an associated probability distribution on independent sets:
\begin{align}
     \cZ_{G,k}(t)
  &:=
  \sum_{I\in\I(G)}\wt_{G,t,k}(I) \\
  \cmu_{G,t,k}(I) &:=
  \frac{\wt_{G,t,k}(I)}{\cZ_{G,k}(t)} \,.
\end{align}

To motivate this specific choice of weight function, observe that the identity
\begin{equation}\label{eq:shift}
  f(m)+km = f(m+k)-f(k)
\end{equation}
implies that multiplication by the constant \(q^{f(k)}\) turns the
component weight into \(t^{|I|}q^{f(m(I)+k)}\). Thus, the quadratically localized measure defined here can be viewed as a quadratic extension of the tilted hard-core law defined in \cite{narang2026slices}.
We will now show that one can sample from the quadratically localized measure efficiently via Glauber dynamics.

\begin{lemma} \label{lem:cancellation} For every \(t>0\), every \(k\in\mathbb Z\), and every feasible occupied pinning \(S\), the
pairwise dependency matrix of $\cmu_{G,t,k}$  satisfies
\begin{equation}\label{eq:psi-bound}
  \Psi_S
  \preceq
  \Id+\frac{t\sigma}{\sqrt q}\diag(r_S)^{-1}.
\end{equation}
Consequently, if \(t\sigma<\sqrt q\), the hypothesis of
Theorem~\ref{thm:chen-general} holds with $\delta=1-\frac{t\sigma}{\sqrt q}$.
\end{lemma}

\begin{proof}
We will write $\wt$ for $\wt_{G,t,k}$. Fix a feasible independent set \(S\) and denote $d:=m(S).$ Let \(U_L\subseteq L\) and \(U_R\subseteq R\) be the vertices that are
available after occupying \(S\), and set
\[
  B:=M[U_L,U_R],
\]
that is, the incidence matrix of $G$ restricted to the rows and columns of available vertices. Now, given some \(u\in U_L\), we compute the one-site ratios under the pinning $S$. In particular
\begin{align*}
  r_L &= \frac{\wt(S \cup \{u\})}{\wt(S)} \\ &= \frac{t^{|S| + 1}q^{f(d + 1) + k \cdot (d + 1)}}{t^{|S|}q^{f(d) + k \cdot d}} \\
  &= t \cdot q^{f(d + 1) - f(d)} \cdot q^k \\
  &= t \cdot q^{k + d + 1/2}
\end{align*}
and by a similar computation for \(v\in U_R\), we have that $r_R = t \cdot q^{1/2 - k -d}$.  While the individual left and right marginal ratios  vary exponentially with $k+d$, their geometric mean is always $\sqrt{r_L r_R } = t \sqrt{q}$. Now, we compute the symmetric dependency matrix $\Psi_S$. Let $u$ and $v$ be two distinct available vertices on the same side. Then, we have

\begin{align*}
\Psi_S(u,v)
&= \frac{\wt(S\cup\{u,v\})\wt(S)}
        {\wt(S\cup\{u\})\wt(S\cup\{v\})}-1 \\
&= \frac{
    \bigl(t^{|S|+2}q^{f(d+2)+k(d+2)}\bigr)
    \bigl(t^{|S|}q^{f(d)+kd}\bigr)}
    {
    \bigl(t^{|S|+1}q^{f(d+1)+k(d+1)}\bigr)^2}
    -1 \\
&= q^{f(d+2)+f(d)-2f(d+1)}-1 \\
&= q-1.
\end{align*}

Now, consider a pair $u \in L$ and $v \in R$. If $(u,v) \in E(G)$, then $\Psi_S(u,v) = -1$. If not, then by a similar computation,

\begin{align*}
\Psi_S(u,v)
&= \frac{
    \bigl(t^{|S|+2}q^{f(d)+kd}\bigr)
    \bigl(t^{|S|}q^{f(d)+kd}\bigr)}
    {
    \bigl(t^{|S|+1}q^{f(d+1)+k(d+1)}\bigr)
    \bigl(t^{|S|+1}q^{f(d-1)+k(d-1)}\bigr)}
    -1 \\
&= q^{2f(d)-f(d+1)-f(d-1)}-1 \\
&= q^{-1}-1.
\end{align*}

Set $a:=\frac{\Delta}{n}=1-q$ and write \(J_L,J_R,J_{LR},J_{RL}\) for the appropriate all-ones matrices. Then, from above, the dependency matrix is exactly
\begin{equation}\label{eq:psi-block}
  \Psi_S
  =
  \begin{pmatrix}
    -a(J_L-\Id_L)
      &(q^{-1}-1)J_{LR}-q^{-1}B\\
    (q^{-1}-1)J_{RL}-q^{-1}B^{\mathsf T}
      &-a(J_R-\Id_R)
  \end{pmatrix}.
\end{equation}
Because $G$ is $\Delta$-regular, we can decompose its biadjacency matrix $M$ in terms of the one-eigenspace and its complement. This gives
\[
  M=\frac{\Delta}{n}J+M_\perp,
  \qquad
  M_\perp\one=M_\perp^{\mathsf T}\one=0,
  \qquad
  \|M_\perp\|_{\op}=\sigma_2(M)\leq\sigma.
\]
Now, set $E_S:=B-\frac{\Delta}{n}J$, where here \(J\) is the
\(|U_L|\times|U_R|\) all-ones matrix. Then, by definition, \(E_S=M_\perp[U_L,U_R]\). As coordinate projections are contractions, we have that $\|E_S\|_{\op}\le\sigma$. Denote
\[
  R_S:=\diag(r_LI_{U_L},r_RI_{U_R}).
\]
Using the fact that $\sqrt{r_Lr_R}=t\sqrt q$ and the cancellation
\[
(q^{-1}-1)J-q^{-1}B=-q^{-1}E_S,
\]
we compute
\begin{align*}
  R_S^{1/2}\Psi_SR_S^{1/2}
  &=
  \begin{pmatrix}
    \sqrt{r_L}\Id_L & 0\\
    0 & \sqrt{r_R}\Id_R
  \end{pmatrix}
  \begin{pmatrix}
    -a(J_L-\Id_L)
      &-q^{-1}E_S\\
    -q^{-1}E_S^{\mathsf T}
      &-a(J_R-\Id_R)
  \end{pmatrix}
  \begin{pmatrix}
    \sqrt{r_L}\Id_L & 0\\
    0 & \sqrt{r_R}\Id_R
  \end{pmatrix} \\
  &=
  \begin{pmatrix}
    -ar_L(J_L-\Id_L)
      &-q^{-1}\sqrt{r_Lr_R}\,E_S\\
    -q^{-1}\sqrt{r_Lr_R}\,E_S^{\mathsf T}
      &-ar_R(J_R-\Id_R)
  \end{pmatrix} \\
  &=
  \begin{pmatrix}
    -ar_LJ_L+ar_L\Id_L
      &-tq^{-1/2}E_S\\
    -tq^{-1/2}E_S^{\mathsf T}
      &-ar_RJ_R+ar_R\Id_R
  \end{pmatrix}.
\end{align*}
Since \(a=1-q\), subtracting gives
\begin{align*}
  R_S^{1/2}\Psi_SR_S^{1/2}-R_S
  &=
  \begin{pmatrix}
    -ar_LJ_L+(a-1)r_L\Id_L
      &-tq^{-1/2}E_S\\
    -tq^{-1/2}E_S^{\mathsf T}
      &-ar_RJ_R+(a-1)r_R\Id_R
  \end{pmatrix} \\
  &=
  \begin{pmatrix}
    -ar_LJ_L-qr_L\Id_L
      &-tq^{-1/2}E_S\\
    -tq^{-1/2}E_S^{\mathsf T}
      &-ar_RJ_R-qr_R\Id_R
  \end{pmatrix}.
\end{align*}

Observe that because $a, r_L, r_R > 0$, the matrices $-ar_LJ_L-qr_L\Id_L$ and $-ar_RJ_R-qr_R\Id_R$ are negative semidefinite. Thus, for any vectors
\(x\in\mathbb R^{U_L}\) and \(y\in\mathbb R^{U_R}\),
\begin{align*}
  \left\langle
    \binom{x}{y},
    (R_S^{1/2}\Psi_SR_S^{1/2}-R_S)
    \binom{x}{y}
  \right\rangle
  &\le
  \frac{2t}{\sqrt q}|x^{\mathsf T}E_Sy|\\
  &\le
  \frac{t\sigma}{\sqrt q}
  \bigl(\|x\|_2^2+\|y\|_2^2\bigr),
\end{align*}
where the last inequality uses the fact that $\|E_S\|_{\op}\leq\sigma$ and
\(2ab\le a^2+b^2\). Therefore
\[
  R_S^{1/2}\Psi_SR_S^{1/2}-R_S
  \preceq
  \frac{t\sigma}{\sqrt q}I.
\]
After rearranging and applying congruence by \(R_S^{-1/2}\), this is
exactly \eqref{eq:psi-bound}.
\end{proof}

\subsection{An exact mixture of localized measures}

We now show how to transform samplers for the measures $\cmu_{G,t,k}$ (as $k$ varies) into a sampler for the hard-core measure. The underlying idea is to use a discrete Gaussian distribution to remove the quadratic exponent in the quadratically localized measure. In particular, define $p_k := \frac{q^{k^2/2}}{\displaystyle\sum_{r\in\mathbb{Z}} q^{r^2/2}}$ so that $\{p_k\}_{k \in \mathbb{Z}}$ defines a probability distribution on the integers. We denote this measure by $p$. Since \(f(r)=f(-r)\), this distribution is symmetric about \(0\). A useful property of this distribution is the following.

\begin{lemma} \label{lem:expect_one}
For all $0<q<1$ and every independent set $I$,
\[
\mathbb{E}_{k\sim p}
\left[
q^{f(m(I))+km(I)}
\right] = 1.
\]
\end{lemma}

\begin{proof} Observe that because $q < 1$, the sum $\sum_{k \in \mathbb{Z}} q^{k^2/2}$ is absolutely convergent. Therefore, we can compute as follows:
\begin{align*}
\mathbb{E}_{k\sim p}
\left[
q^{f(m(I))+km(I)}
\right]
&=
\sum_{k\in\mathbb{Z}}
p_k q^{f(m)+km}
\\
&=
\frac{1}{\sum_{r \in \mathbb{Z}} q^{f(r)}}
\sum_{k\in\mathbb{Z}}
q^{f(k)}q^{f(m)+km}
\\
&=
\frac{1}{\sum_{r \in \mathbb{Z}} q^{f(r)}} \cdot
\sum_{k\in\mathbb{Z}}
q^{f(k+m)}
\\
&=1.
\end{align*}
Here, the third equality uses the fact that $
f(k+m)=f(k)+f(m)+km$ and the last equality follows from the change of variables \(r=k+m\).
\end{proof}

With the above lemma in hand, we can now state the key observation which enables us to relate the quadratically localized partition function to the hard-core model's partition function.

\begin{align} \label{eq:inf_sum}
Z_G(t)
&= \sum_{I\in\I(G)} t^{|I|} \nonumber \\
&= \sum_{I\in\I(G)}
t^{|I|}
\mathbb E_{k\sim p}
\left[
q^{f(m(I))+km(I)}
\right] \nonumber \\
&= \mathbb E_{k\sim p}
\left[
\sum_{I\in\I(G)}
t^{|I|}q^{f(m(I))+km(I)}
\right] \nonumber \\
&= \mathbb E_{k\sim p}\left[\cZ_{G,k}(t)\right] \nonumber \\
&= \sum_{k\in\mathbb Z} p_k \cZ_{G,k}(t),
\end{align}
where the second equality follows from Lemma \ref{lem:expect_one}. Thus, it remains to show that the infinite sum \eqref{eq:inf_sum} can be approximated via a truncation at a polynomial threshold with sufficiently good estimators.

\begin{lemma}
\label{lem:trunc-mixture}
Fix \(0<\varepsilon<1\), and set $
K
:=
2n+\left\lceil \log_2\frac{2}{\varepsilon}\right\rceil.$
Then, for every \(t>0\),
\[
(1-\varepsilon)Z_G(t)
\leq
\sum_{k=-K}^{K}p_k\cZ_{G,k}(t)
\leq
Z_G(t).
\]
\end{lemma}

\begin{proof}
Set $s:=\left\lceil \log_2\frac{2}{\varepsilon}\right\rceil$ so that $K=2n+s.$
Additionally, define
\[
\alpha_k(t):=\frac{p_k\cZ_{G,k}(t)}{Z_G(t)}.
\]
Using the fact that $f(k+m)=f(k)+f(m)+km$, we have, for every independent set \(I\), $p_kq^{f(m(I))+km(I)}
=
p_{k+m(I)}.$ Consequently,
\begin{align*}
\alpha_k(t)
&=
\frac{p_k}{Z_G(t)}
\sum_{I\in\I(G)}
t^{|I|}q^{f(m(I))+km(I)}\\
&=
\sum_{I\in\I(G)}
\mu_{G,t}(I)\,p_{k+m(I)}.
\end{align*}
Therefore, we can write
\begin{align*}
\sum_{|k|>K}\alpha_k(t)
&=
\mathbb E_{I\sim\mu_{G,t}}
\left[
\sum_{|k|>K}p_{k+m(I)}
\right].
\end{align*}
Since \(|m(I)|\leq n\), the condition \(|k|>2n+s\) implies $|k+m(I)|>n+s.$ It follows that
\[
\sum_{|k|>K}p_{k+m(I)}
\leq
\sum_{|r|>n+s}p_r.
\]
So, all that remains is to bound the tail of \(p\). For every \(r\geq n\),
\[
\frac{p_{r+1}}{p_r}
=
q^{r+1/2}
\leq
q^n
=
\left(1-\frac{\Delta}{n}\right)^n
\leq
e^{-\Delta}
\leq
\frac12.
\]
Hence
\[
\sum_{r>n+s}p_r
\leq
p_n\sum_{j\geq1}2^{-(s+j)}
\leq
2^{-s}.
\]
Moreover, \(f(r)=f(-r)\), and therefore \(p_r=p_{-r}\). Thus the
left and right tails are equal, and
\[
\sum_{|r|>n+s}p_r
\leq
2^{1-s}
\leq
\varepsilon.
\]
Therefore, we have that $\sum_{|k|>K}\alpha_k(t)\leq\varepsilon$ and multiplying by \(Z_G(t)\) proves the result.
\end{proof}

\subsection{The Algorithms}
\label{sec:reconstruction-algorithm}

We now combine the quadratically localized sampler with the mixture identity to state our algorithms and provide the proof of Theorem \ref{thm:main}. We begin by isolating the two algorithmic ingredients (Glauber dynamics and simulated annealing) that will be used by both the sampler and the counting algorithm.

For Glauber dynamics, fix an auxiliary accuracy
\(\varepsilon_K\in(0,1)\) and let
\[
K=2n+\left\lceil\log_2\frac{2}{\varepsilon_K}\right\rceil
\]
as in Lemma \ref{lem:trunc-mixture}. For \(|k|\leq K\), \(0<t\leq\lambda\), and \(\tau>0\), let
\(\textnormal{\textsc{QuadraticallyLocalizedSampler}}(t,k,\tau)\) denote Glauber dynamics for
\(\cmu_{G,t,k}\), started from the empty independent set and run long enough to have total variation error at most \(\tau\).

\begin{lemma}
\label{lem:quadratically-localized-sampler}
For fixed \(\Delta\) and \(\xi\), \(\textnormal{\textsc{QuadraticallyLocalizedSampler}}(t,k,\tau)\) runs in time polynomial in
\(n\), \(1/\varepsilon_K\), \(1+t\), \(1+|\log t|\), and
\(\log(1/\tau)\).
\end{lemma}

\begin{proof}
Fix a feasible occupied pinning \(S\), and write \(d=m(S)\). The
one-site ratios computed in the proof of
Lemma~\ref{lem:cancellation} are
\[
r_L=tq^{k+d+1/2},
\qquad
r_R=tq^{1/2-k-d}.
\]
Since \(|d|\leq n\) and \(|k|\leq K\), every available vertex \(v\)
satisfies
\begin{equation}
\label{eq:ratio-range}
tq^{K+n+1}
\leq
r_S(v)
\leq
tq^{-(K+n)}.
\end{equation}
In particular, this gives $r_{\max}\leq tq^{-(K+n)}$, where $r_{\max}$ is defined as in Theorem \ref{thm:chen-general}.
Since \(\Delta/n\leq\xi\),
\[
-\log q
\leq
\frac{\Delta}{n-\Delta}
\leq
\frac{\Delta}{n(1-\xi)}.
\]
The definition of \(K\) therefore gives
\[
q^{-(K+n)}
\leq
\exp\!\left[
\frac{\Delta(K+n)}{n(1-\xi)}
\right]
\leq
(1/\varepsilon_K)^{O_{\Delta,\xi}(1)}.
\]
Every independent set can be constructed from \(\varnothing\) by
successively adding its vertices. Applying the lower bound in
\eqref{eq:ratio-range} at each addition gives
\[
\wt_{G,t,k}(I)
\geq
\min\left\{
1,
\bigl(tq^{K+n+1}\bigr)^{2n}
\right\}.
\]
Similarly, the upper bound in \eqref{eq:ratio-range} gives
\[
\cZ_{G,k}(t)
\leq
\sum_{A\subseteq L\sqcup R}
\bigl(tq^{-(K+n)}\bigr)^{|A|}
=
\bigl(1+tq^{-(K+n)}\bigr)^{2n}.
\]
Consequently, every state of the quadratically localized chain has probability at least
\[
\mu_*
:=
\frac{
\min\left\{1,\bigl(tq^{K+n+1}\bigr)^{2n}\right\}
}{
\bigl(1+tq^{-(K+n)}\bigr)^{2n}
}.
\]
By Lemma~\ref{lem:cancellation}, the hypothesis of
Theorem~\ref{thm:chen-general} holds with $\delta_t
=
1-\frac{t\sigma}{\sqrt q}.$ Since \(t\leq\lambda\), we may use the uniform lower bound
$\delta_t
\geq
\delta
:=
1-\frac{\lambda\sigma}{\sqrt q}.$ Moreover, \(\lambda\sigma\leq1-\xi\) and \(q\geq1-\xi\), so
\[
\delta
\geq
1-\frac{1-\xi}{\sqrt q}
\geq
1-\sqrt{1-\xi}.
\]
Let \(C_{\rm mix}>0\) be a sufficiently large universal constant so
that the mixing-time estimate below is valid. After one
\(T_{\mathrm{mix}}\)-step block, the maximal pairwise total variation
distance is at most \(1/2\). Its standard submultiplicativity therefore
implies that \(r\) consecutive blocks have error at most \(2^{-r}\).
Accordingly, we approximately sample from $\cmu_{G,t,k}$ by starting
Glauber dynamics from the empty independent set and running it for
\begin{equation}
\label{eq:quadratically-localized-sampler-runtime}
\begin{split}
T(t,K,\tau)
:=
\Bigg\lceil
&C_{\rm mix}
\frac{2n\bigl(1+tq^{-(K+n)}\bigr)}{\delta} \cdot
\log\left(
\frac{
\bigl(1+tq^{-(K+n)}\bigr)^{2n}
}{
\min\left\{1,\bigl(tq^{K+n+1}\bigr)^{2n}\right\}
}
\right)
\log\frac{2}{\tau}
\Bigg\rceil
\end{split}
\end{equation}
steps. By Theorem~\ref{thm:chen-general}, the distribution of the
final state is within total variation distance \(\tau\) of
\(\cmu_{G,t,k}\). The preceding bound on \(q^{-(K+n)}\), together
with \(\delta\geq1-\sqrt{1-\xi}\), shows directly from
\eqref{eq:quadratically-localized-sampler-runtime} that this running time is
polynomial in the parameters stated in the lemma.
\end{proof}

We next describe the simulated annealing procedure used in both
algorithms. Let \(C>0\) be a sufficiently large universal constant. We denote the output of the following algorithm (Algorithm~\ref{alg:annealed-estimate}) by $\textnormal{\textsc{AnnealedEstimate}}(k,\eta,\zeta).$

\begin{algorithm}[H]
\caption{Simulated annealing estimator for \(\cZ_{G,k}(\lambda)\)}
\label{alg:annealed-estimate}
\begin{algorithmic}[1]
\Require An index \(|k|\leq K\), an accuracy parameter
\(\eta\in(0,1)\), and a failure probability \(\zeta\in(0,1)\)
\Ensure An estimate \(\widehat{\cZ}_k\) of \(\cZ_{G,k}(\lambda)\)
\State Set \(t_0=\min\{\lambda,\eta/(128nq^{-(K+n)})\}\),
\(s=\left\lceil \log(\lambda/t_0)/\log(1+1/(2n))\right\rceil\),
\(\gamma=\eta/(64\max\{s,1\})\), and
\(N=\left\lceil C\gamma^{-2}\log(2\max\{s,1\}/\zeta)\right\rceil\)
\State Set the annealing schedule as follows
\[
t_i=t_0\left(1+\frac1{2n}\right)^i
\quad (0\leq i<s),
\qquad
t_s=\lambda
\]
\State Set \(\widehat{\cZ}_k\gets1\)
\For{\(i=1,\ldots,s\)}
    \For{\(j=1,\ldots,N\)}
        \State
        \(I_{i,j}\gets
        \textnormal{\textsc{QuadraticallyLocalizedSampler}}
        (t_i,k,\gamma/16)\)
    \EndFor
    \State Set
    \[
    \widehat R_{i,k}
    =
    \frac1N\sum_{j=1}^{N}
    \left(\frac{t_{i-1}}{t_i}\right)^{|I_{i,j}|}
    \]
    \State
    \(\widehat{\cZ}_k
    \gets
    \widehat{\cZ}_k/\widehat R_{i,k}\)
\EndFor
\State \Return \(\widehat{\cZ}_k\)
\end{algorithmic}
\end{algorithm}

\begin{lemma}
\label{lem:annealed-estimate}
For every \(|k|\leq K\),
\[
    \mathbb P\left(
    (1-\eta)\cZ_{G,k}(\lambda)
    \leq \textnormal{\textsc{AnnealedEstimate}}(k,\eta,\zeta)
    \leq (1+\eta)\cZ_{G,k}(\lambda)
    \right)
    \geq1-\zeta.
\]
Moreover, $\textnormal{\textsc{AnnealedEstimate}}(k,\eta,\zeta)$ can be computed in time
polynomial in \(n\), \(1/\varepsilon_K\), \(1/\eta\), and \(1/\zeta\).
\end{lemma}

\begin{proof}
If \(s=0\), then \(t_0=\lambda\), the algorithm returns \(1\), and the
estimate for \(\cZ_{G,k}(t_0)\) below directly gives the claimed
accuracy. Thus, suppose that \(s\geq1\).

For each \(i=1,\ldots,s\), we have
\begin{align*}
R_{i,k}
&:=
\frac{\cZ_{G,k}(t_{i-1})}{\cZ_{G,k}(t_i)} \\
&=
\frac{1}{\cZ_{G,k}(t_i)}
\sum_{I\in\I(G)}
t_{i-1}^{|I|}
q^{f(m(I))+km(I)} \\
&=
\mathbb E_{I\sim\cmu_{G,t_i,k}}
\left[
\left(\frac{t_{i-1}}{t_i}\right)^{|I|}
\right].
\end{align*}

By the choice of $s$, we have that $\frac{t_i}{t_{i-1}}
\leq
1+\frac1{2n}$ for every \(i=1,\ldots,s\), including the final step \(t_s=\lambda\).
Since \(|I|\leq2n\),
\[\frac{1}{4} \le
\left(1+\frac1{2n}\right)^{-2n}
\leq
\left(\frac{t_{i-1}}{t_i}\right)^{|I|}
\leq
1
\]
and so $R_{i,k}\geq\frac14$. By Lemma \ref{lem:quadratically-localized-sampler}, each sample \(I_{i,j}\) has distribution within total variation
distance \(\gamma/16\) of \(\cmu_{G,t_i,k}\). Since the function $I\longmapsto
\left(\frac{t_{i-1}}{t_i}\right)^{|I|}$ takes values in \([0,1]\), it follows that
\[
\left|
\mathbb E[\widehat R_{i,k}]-R_{i,k}
\right|
\leq
\frac{\gamma}{16}.
\]
The \(N\) summands defining \(\widehat R_{i,k}\) are independent and
lie in \([0,1]\). Hence, by Hoeffding's inequality,
\begin{align*}
\mathbb P\left(
\left|
\widehat R_{i,k}-\mathbb E[\widehat R_{i,k}]
\right|
>
\frac{3\gamma}{16}
\right)
&\leq
2\exp\left(
-2N\left(\frac{3\gamma}{16}\right)^2
\right) \\
&=
2\exp\left(-\frac{9N\gamma^2}{128}\right) \\
&\leq
\frac{\zeta}{\max\{s,1\}},
\end{align*}
where the last inequality follows from the definition of \(N\) and a
sufficiently large choice of \(C\). Summing over \(i=1,\ldots,s\), we find that, with probability at least
\(1-\zeta\), we have that $\left|
\widehat R_{i,k}-R_{i,k}
\right|
\leq
\frac{\gamma}{4}$ simultaneously for every \(i\). Since \(R_{i,k}\geq1/4\), this implies
\begin{equation}
\label{eq:annealing-ratios}
(1-\gamma)R_{i,k}
\leq
\widehat R_{i,k}
\leq
(1+\gamma)R_{i,k}
\end{equation}
with probability at least $1 - \zeta$.

At fugacity \(t_0\), every one-site ratio is at most
\(t_0q^{-(K+n)}\). Hence
\[
1
\leq
\cZ_{G,k}(t_0)
\leq
\sum_{A\subseteq L\sqcup R}
\bigl(t_0q^{-(K+n)}\bigr)^{|A|}
=
\bigl(1+t_0q^{-(K+n)}\bigr)^{2n}
\leq
e^{2nt_0q^{-(K+n)}}
\leq
e^{\eta/64}.
\]
Moreover,
\[
\prod_{i=1}^{s}R_{i,k}
=
\frac{\cZ_{G,k}(t_0)}{\cZ_{G,k}(\lambda)},
\qquad
\frac{\widehat{\cZ}_k}{\cZ_{G,k}(\lambda)}
=
\frac1{\cZ_{G,k}(t_0)}
\prod_{i=1}^{s}\frac{R_{i,k}}{\widehat R_{i,k}}.
\]
Thus, on the event \eqref{eq:annealing-ratios}, we have
\begin{align*}
\frac{\widehat{\cZ}_k}{\cZ_{G,k}(\lambda)}
=
\frac1{\cZ_{G,k}(t_0)}
\prod_{i=1}^{s}\frac{R_{i,k}}{\widehat R_{i,k}} \leq
(1-\gamma)^{-s}
\leq
e^{2s\gamma}
\leq
e^{\eta/32}
\leq
1+\eta.
\end{align*}
For the lower bound (on the same event),
\begin{align*}
\frac{\widehat{\cZ}_k}{\cZ_{G,k}(\lambda)}=
\frac1{\cZ_{G,k}(t_0)}
\prod_{i=1}^{s}\frac{R_{i,k}}{\widehat R_{i,k}} \geq
e^{-\eta/64}(1+\gamma)^{-s}
\geq
e^{-\eta/64-s\gamma}
\geq
e^{-\eta/32}
\geq
1-\eta.
\end{align*}
giving the algorithmic guarantee. As the annealing procedure makes \(sN\) calls to
\(\textnormal{\textsc{QuadraticallyLocalizedSampler}}\), each with accuracy
\(\gamma/16\). Lemma~\ref{lem:quadratically-localized-sampler} gives that the algorithm runs in polynomial time.
\end{proof}

We now use $\textnormal{\textsc{QuadraticallyLocalizedSampler}}$ and
$\textnormal{\textsc{AnnealedEstimate}}$ as subroutines to obtain the efficient sampling algorithm and FPRAS in Theorem \ref{thm:main}.

\subsubsection{An Efficient Sampler for \(\mu_{G,\lambda}\)}

For every independent set \(I\), Lemma~\ref{lem:expect_one} gives
\begin{align}
    \mu_{G,\lambda}(I)
    &=\sum_{k\in\mathbb Z}
    \frac{p_k\cZ_{G,k}(\lambda)}{Z_G(\lambda)}\cmu_{G,\lambda,k}(I).
    \label{eq:measure-mixture}
\end{align}
Thus the hard-core measure is a mixture of quadratically localized measures. Algorithm~\ref{alg:sampler} turns this into a sampler by first approximating the mixture weights and then sampling from the selected quadratically localized measure.

\begin{algorithm}[H]
\caption{Approximate sampler for the hard-core measure}
\label{alg:sampler}
\begin{algorithmic}[1]
\Require A \(\Delta\)-regular bipartite graph \(G=(L\sqcup R,E)\) with second biadjacency singular value at most \(\sigma\), a fixed constant \(\xi\in(0,1)\), a fugacity \(0<\lambda\leq(1-\xi)/\sigma\), and an accuracy parameter \(\varepsilon\in(0,1)\)
\Ensure A sample from \(\mu_{G,\lambda}\), up to total variation error \(\varepsilon\)
\If{\(n<\Delta/\xi\)}
    \State Sample exactly by enumerating \(\I(G)\), and return
\EndIf
\State Set \(\varepsilon_0=\varepsilon/20\),
\(\varepsilon_K=\varepsilon_0/4\), and
\(K=2n+\lceil\log_2(8/\varepsilon_0)\rceil\)
\For{\(k=-K,\ldots,K\)}
    \State \(\widehat{\cZ}_k\gets\textnormal{\textsc{AnnealedEstimate}}(k,\varepsilon_0,\varepsilon_0/(2K+1))\)
    \State \(\widehat A_k\gets q^{f(k)}\widehat{\cZ}_k\)
\EndFor
\State Draw \(\widehat k\in\{-K,\ldots,K\}\) with probability \(\widehat A_{\widehat k}/\sum_{j=-K}^{K}\widehat A_j\)
\State \Return \(\textnormal{\textsc{QuadraticallyLocalizedSampler}}(\lambda,\widehat k,\varepsilon_0)\)
\end{algorithmic}
\end{algorithm}

The following lemma gives correctness and runtime guarantees for Algorithm \ref{alg:sampler}.

\begin{lemma}
\label{lem:sampler}
Algorithm~\ref{alg:sampler} runs in time polynomial in \(n\) and \(1/\varepsilon\), and its
output law \(\widehat\mu\) satisfies
$\TV(\widehat\mu,\mu_{G,\lambda})<\varepsilon.$
\end{lemma}

\begin{proof}
Let $\alpha_k:=\frac{p_k\cZ_{G,k}(\lambda)}{Z_G(\lambda)}.$ By \eqref{eq:measure-mixture}, we have that $\mu_{G,\lambda}=\sum_{k\in\mathbb Z}\alpha_k\cmu_{G,\lambda,k}$. By Lemma~\ref{lem:trunc-mixture}, applied with \(\varepsilon=\varepsilon_0/4\),
\[
    \sum_{|k|>K}\alpha_k\leq\frac{\varepsilon_0}{4}.
\]
Since \(p_k\) is proportional to \(q^{f(k)}\), the exact
mixture weights on \([-K,K]\) are proportional to
\[
    A_k:=q^{f(k)}\cZ_{G,k}(\lambda).
\]

By Lemma~\ref{lem:annealed-estimate} and a union bound, with probability at least
\(1-\varepsilon_0\),
\[
    (1-\varepsilon_0)A_k\leq\widehat A_k\leq(1+\varepsilon_0)A_k
\]
simultaneously for every \(|k|\leq K\). On this event, normalization changes the resulting
distribution on \([-K,K]\) by total variation distance at most $\frac{\varepsilon_0}{1-\varepsilon_0}\leq2\varepsilon_0.$ The final call to \(\textnormal{\textsc{QuadraticallyLocalizedSampler}}\) contributes at most another
\(\varepsilon_0\). Accounting also for the exceptional estimation event, convexity of total
variation distance gives
\[
    \TV(\widehat\mu,\mu_{G,\lambda})
    \leq \frac{\varepsilon_0}{4}+2\varepsilon_0+\varepsilon_0+\varepsilon_0
    <\varepsilon.
\]
Polynomial running time follows from Lemmas~\ref{lem:quadratically-localized-sampler} and
\ref{lem:annealed-estimate}.
\end{proof}

\subsubsection{FPRAS for \(Z_G(\lambda)\)}

The counting algorithm uses the same annealing estimates. The only additional step is to
approximate the normalizing constant
\[
    P(q):=\sum_{r\in\mathbb Z}q^{f(r)}
\]
in the definition \(p_k=q^{f(k)}/P(q)\).

\begin{algorithm}[H]
\caption{FPRAS for the hard-core partition function}
\label{alg:fpras}
\begin{algorithmic}[1]
\Require A \(\Delta\)-regular bipartite graph \(G=(L\sqcup R,E)\) with second biadjacency singular value at most \(\sigma\), a fixed constant \(\xi\in(0,1)\), a fugacity \(0<\lambda\leq(1-\xi)/\sigma\), and an accuracy parameter \(\varepsilon\in(0,1)\)
\Ensure An \(\varepsilon\)-relative approximation to \(Z_G(\lambda)\)
\If{\(n<\Delta/\xi\)}
    \State Compute \(Z_G(\lambda)\) exactly by enumerating \(\I(G)\), and return
\EndIf
\State Set \(\varepsilon_0=\varepsilon/20\),
\(\varepsilon_K=\varepsilon_0/4\), and
\(K=2n+\lceil\log_2(8/\varepsilon_0)\rceil\)
\For{\(k=-K,\ldots,K\)}
    \State \(\widehat{\cZ}_k\gets\textnormal{\textsc{AnnealedEstimate}}(k,\varepsilon_0,1/(10(2K+1)))\)
\EndFor
\State Set \(B=\lceil2n+2\log(64/\varepsilon_0)\rceil\) and \(\widehat P=\sum_{r=-B}^{B}q^{f(r)}\)
\State \Return \(\widehat Z=\widehat P^{-1}\sum_{k=-K}^{K}q^{f(k)}\widehat{\cZ}_k\)
\end{algorithmic}
\end{algorithm}

\begin{lemma}
\label{lem:fpras}
Algorithm~\ref{alg:fpras} is an FPRAS for \(Z_G(\lambda)\).
\end{lemma}

\begin{proof}
As in the sampling argument, Lemma~\ref{lem:annealed-estimate} and a union bound imply that, with probability at least \(9/10\),
\begin{equation}
    (1-\varepsilon_0)\cZ_{G,k}(\lambda)
    \leq \widehat{\cZ}_k
    \leq (1+\varepsilon_0)\cZ_{G,k}(\lambda)
    \label{eq:all-Zk}
\end{equation}
simultaneously for all \(|k|\leq K\). Set
\[
    Z^{(K)}
    :=\sum_{k=-K}^{K}p_k\cZ_{G,k}(\lambda),
    \qquad
    \widetilde Z^{(K)}
    :=\sum_{k=-K}^{K}p_k\widehat{\cZ}_k.
\]
On the event in \eqref{eq:all-Zk},
\begin{equation}
    (1-\varepsilon_0)Z^{(K)}
    \leq \widetilde Z^{(K)}
    \leq (1+\varepsilon_0)Z^{(K)}.
    \label{eq:truncated-count-estimate}
\end{equation}
By Lemma~\ref{lem:trunc-mixture}, applied with \(\varepsilon=\varepsilon_0/4\),
\begin{equation}
    (1-\varepsilon_0/4)Z_G(\lambda)
    \leq Z^{(K)}
    \leq Z_G(\lambda).
    \label{eq:Z-truncation}
\end{equation}

It remains only to account for the approximation of \(P(q)\). Since \(q\leq e^{-1/n}\), for every \(r\geq B\) we have $q^{f(r+1)}/q^{f(r)}=q^{r+1/2}\leq e^{-B/n}\leq e^{-2}$. Together with $f(r)=f(-r)$, summing the two geometric tails gives
\begin{align}
0\leq P(q)-\widehat P
&=\sum_{r>B}q^{f(r)}+\sum_{r<-B}q^{f(r)} \nonumber\\
&\leq 2\sum_{r>B}q^{f(r)} \nonumber\\
&\leq \frac{2q^{f(B+1)}}{1-e^{-2}} \nonumber\\
&\leq 4\exp\!\left(-\frac{(B+1)^2}{2n}\right) \nonumber\\
&\leq 4\exp\!\left(-\frac{B^2}{2n}\right) \nonumber\\
&\leq 4\exp\!\left(-2n-2\log\frac{64}{\varepsilon_0}\right)
\leq\varepsilon_0.
\end{align}
As \(P(q)\geq1\) from the term indexed by \(0\), this implies
\begin{equation}
    1\leq\frac{P(q)}{\widehat P}\leq1+\varepsilon_0.
    \label{eq:P-approx}
\end{equation}
Using \(q^{f(k)}/\widehat P=(P(q)/\widehat P)p_k\), we have that $\widehat Z =\frac{P(q)}{\widehat P}\,\widetilde Z^{(K)}.$ Combining \eqref{eq:truncated-count-estimate},
\eqref{eq:Z-truncation}, and \eqref{eq:P-approx}, we obtain
\[
    (1-\varepsilon_0)(1-\varepsilon_0/4)Z_G(\lambda)
    \leq \widehat Z
    \leq (1+\varepsilon_0)^2Z_G(\lambda).
\]
Since \(\varepsilon_0=\varepsilon/20\), these bounds imply
\[
    (1-\varepsilon)Z_G(\lambda)
    \leq \widehat Z
    \leq (1+\varepsilon)Z_G(\lambda).
\]

Finally, \(K\) and \(B\) are polynomial in \(n\) and
\(\log(1/\varepsilon)\), and every annealing estimate runs in polynomial
time by Lemma~\ref{lem:annealed-estimate}. Thus the algorithm is an
FPRAS.
\end{proof}

\section{Proof of Theorem \ref{thm:cluster-expansion}}
\label{sec:cluster-proof}

Throughout this section, fix $0\leq\sigma<\Delta$ and let
$G\in\mathcal G_{\Delta,\Delta}^{\sigma}$ as in Theorem~\ref{thm:cluster-expansion}.

The proof of Theorem \ref{thm:cluster-expansion} combines the existing
polymer algorithm of \cite{JenssenKeevashPerkins2020} with a sharper
spectral analysis. The key new ingredient is a verification and application of the
local- and global-expansion hypotheses of the recent Hadas--Peled
long-range-order criterion \cite{HP2026} directly from the singular
spectrum. We then construct the two polymer models, verify the
Koteck\'y--Preiss criterion, and derive the counting and sampling
algorithms.

\subsection{Long-range order from the singular spectrum}

Let
\[
A=
\begin{pmatrix}
0&M\\ M^{\mathsf T}&0
\end{pmatrix}
\]
be the adjacency matrix of $G$, and denote $P=A/\Delta$. A simple observation is the following.

\begin{lemma}
\label{lem:return-probability}
For every $v\in V(G)$ and every integer $j\geq1$,
\begin{equation}
P^{2j}(v,v)
\leq
\frac1n+
\left(1-\frac1n\right)\left(\frac{\sigma}{\Delta}\right)^{2j}
\leq
\frac1n+\left(\frac{\sigma}{\Delta}\right)^{2j}.
\label{eq:return-probability}
\end{equation}
Moreover, $P^{2j+1}(v,v)=0$ for every $j\geq0$.
\end{lemma}

\begin{proof}
Since $\sigma_2(M)\leq\sigma<\Delta$, the graph $G$ is connected.
Since $P$ is real and symmetric, it has an orthonormal eigenbasis. Its
two trivial eigenvalues are $1$ and $-1$, with normalized eigenvectors
\[
\psi_+
:=
\frac{1}{\sqrt{2n}}\one
\qquad\text{and}\qquad
\psi_-
:=
\frac{1}{\sqrt{2n}}(\one_L-\one_R),
\]
respectively. Extend these to an orthonormal eigenbasis
$\psi_+,\psi_-,\psi_3,\ldots,\psi_{2n}$, and let
$\theta_3,\ldots,\theta_{2n}$ be the corresponding remaining
eigenvalues. Since the eigenvalues of the adjacency matrix of $G$ are
the positive and negative singular values of its biadjacency matrix, we have that $
|\theta_i|\leq\frac{\sigma}{\Delta}$ for $3\leq i\leq2n$.

Let $e_v$ denote the standard basis vector at $v$. The spectral
decomposition of $P^{2j}$ gives
\begin{align*}
P^{2j}(v,v)
&=
\langle e_v,P^{2j}e_v\rangle\\
&=
\langle e_v,\psi_+\rangle^2
+
\langle e_v,\psi_-\rangle^2
+
\sum_{i=3}^{2n}
\theta_i^{2j}\langle e_v,\psi_i\rangle^2\\
&\leq
\frac1n
+
\left(\frac{\sigma}{\Delta}\right)^{2j}
\sum_{i=3}^{2n}\langle e_v,\psi_i\rangle^2\\
&=
\frac1n
+
\left(\frac{\sigma}{\Delta}\right)^{2j}
\left(
\|e_v\|_2^2
-
\langle e_v,\psi_+\rangle^2
-
\langle e_v,\psi_-\rangle^2
\right)\\
&=
\frac1n
+
\left(1-\frac1n\right)
\left(\frac{\sigma}{\Delta}\right)^{2j}\\
&\leq
\frac1n+
\left(\frac{\sigma}{\Delta}\right)^{2j}\,,
\end{align*}
proving the inequality. Finally, a walk on a bipartite graph alternates between the two bipartition classes, so it cannot return to
its starting vertex after an odd number of steps.
\end{proof}

We next show that, at high enough fugacity, independent sets with
substantial occupancy on both bipartition classes carry negligible
total weight. To do this we apply Theorem \ref{prop:HP}.

Let $r = \frac{|V(G)|}{2\Delta}$ and let
\begin{equation}
\mathcal B
:=
\left\{
I\in\I(G):
|I\cap L|>r
\quad\text{and}\quad
|I\cap R|>r
\right\}.
\label{eq:def-bad-event}
\end{equation}

\begin{lemma}
\label{prop:spectral-LRO}
There are absolute constants $C_0,c_0>0$ such that the following holds.
Suppose $n\geq\Delta^2$ and
\begin{equation}
\lambda
\geq
\exp\!\left[
C_0\left(
\frac1\Delta+
\frac{\sigma^2}{\Delta^2-\sigma^2}
\right)\log(\e\Delta)
\right]-1.
\label{eq:LRO-lambda-condition}
\end{equation}
Then
\begin{equation}
\sum_{I\in\mathcal B}\lambda^{|I|}
\leq
(1+\lambda)^n
\exp\!\left[
-c_0\frac{\Delta-\sigma}{\Delta^2}\,
 n\log(1+\lambda)
\right].
\label{eq:spectral-LRO-conclusion}
\end{equation}
\end{lemma}

\begin{proof}
Write $\ell:=\log(1+\lambda)$ so that condition \eqref{eq:LRO-lambda-condition} is exactly
\begin{equation}
\ell
\geq
C_0\left(
\frac1\Delta+
\frac{\sigma^2}{\Delta^2-\sigma^2}
\right)\log(\e\Delta).
\label{eq:ell-alpha-L}
\end{equation}
We apply Theorem \ref{prop:HP} with the cutoff $r$. We split into the following two cases.

\smallskip
\noindent
\emph{Case 1: $\sigma\leq\Delta/2$.}
Use the random-walk certificate
\eqref{eq:HP-random-walk-certificate}. By
Lemma \ref{lem:return-probability} and since $n\geq\Delta^2$,
\begin{align}
\sum_{j=1}^{\lfloor(\Delta-1)/2\rfloor}P^{2j}(v,v)
&\leq
\frac{\Delta}{2n}+
\sum_{j\geq1}\left(\frac{\sigma}{\Delta}\right)^{2j}
\nonumber\\
&\leq
\frac1{2\Delta}+
\frac{\sigma^2}{\Delta^2-\sigma^2}.
\label{eq:return-sum}
\end{align}
Therefore the local expansion condition holds with constants $M_{\mathrm{LE}}=\Delta$ and
\begin{equation}
\frac{C_{\mathrm{LE}}}\Delta
=
\frac{3}{2\Delta}+
\frac{\sigma^2}{\Delta^2-\sigma^2}
\leq
\frac32\left(
\frac1\Delta+
\frac{\sigma^2}{\Delta^2-\sigma^2}
\right).
\label{eq:LE-case1}
\end{equation}
By Lemma \ref{lem:cheeger}, $h(G)\geq\Delta/4.$ Therefore, it follows that
\begin{align}
\log\!\left(
C_{\mathrm{HP}}\frac\Delta{h(G)}\frac{|V(G)|}{r}
\right)
&\leq
\log(C_1\Delta)
\leq C_2\log(\e\Delta),
\label{eq:HP-case1-first}\\
\frac\Delta{h(G)}\frac{|V(G)|}{rM_{\mathrm{LE}}}
&\leq8.
\label{eq:HP-case1-second}
\end{align}
Thus the right-hand side of \eqref{eq:HP-hypothesis} is at most
$C_3(1/\Delta+\sigma^2/(\Delta^2-\sigma^2))\log(\e\Delta)$.

\smallskip
\noindent
\emph{Case 2: $\sigma>\Delta/2$.}
Use the universal one-edge certificate
$C_{\mathrm{LE}}=M_{\mathrm{LE}}=1$.
Then Lemma \ref{lem:cheeger} gives $h(G)\geq(\Delta-\sigma)/2$. Hence
\begin{align}
\frac1\Delta
\log\!\left(
C_{\mathrm{HP}}\frac\Delta{h(G)}\frac{|V(G)|}{r}
\right)
&\leq
\frac1\Delta\log\!\left(
\frac{C_4\Delta}{1-\sigma/\Delta}
\right),
\label{eq:HP-case2-first}\\
\frac1\Delta
\frac\Delta{h(G)}\frac{|V(G)|}{rM_{\mathrm{LE}}}
&\leq
\frac{4}{1-\sigma/\Delta}.
\label{eq:HP-case2-second}
\end{align}
The first term is also $O((1-\sigma/\Delta)^{-1})$ by
$\log x\leq x$. On the other hand, because $\sigma>\Delta/2$,
\begin{equation}
\frac{(\sigma/\Delta)^2}
{(1-\sigma/\Delta)(1+\sigma/\Delta)}
\geq
\frac{1}{6(1-\sigma/\Delta)}.
\label{eq:theta-controls-gap}
\end{equation}
Therefore the right-hand side of \eqref{eq:HP-hypothesis} is at most
$C_5(1/\Delta+\sigma^2/(\Delta^2-\sigma^2))$, and hence at most
$C_5(1/\Delta+\sigma^2/(\Delta^2-\sigma^2))\log(\e\Delta)$.

Choosing $C_0$ larger than the constants in both cases makes
\eqref{eq:ell-alpha-L} imply \eqref{eq:HP-hypothesis}. Applying
Theorem~\ref{prop:HP}, then using
Lemma \ref{lem:cheeger} and $r=n/\Delta$, yields
\begin{align*}
\sum_{I\in\mathcal B}\lambda^{|I|}
&\leq
(1+\lambda)^n
\exp\!\left[-c_{\mathrm{HP}}
\frac{h(G)}\Delta\frac n\Delta\ell\right]\\
&\leq
(1+\lambda)^n
\exp\!\left[-c_0
\frac{\Delta-\sigma}{\Delta^2}n\ell\right],
\end{align*}
which is \eqref{eq:spectral-LRO-conclusion}.
\end{proof}

\subsection{Approximating the Hard-Core Partition Function via Polymer Models}

Having established that there are two dominant phases, we now define the polymer models which capture these. Let $H_L$ be the graph on vertex set $L$ in which two distinct vertices are adjacent when they have a common neighbor in $G$ and define $H_R$
analogously. By construction, each of $H_L,H_R$ has maximum degree at most
\begin{equation}
\Delta(\Delta-1)\leq\Delta^2.
\label{eq:square-degree}
\end{equation}
Note that $H_L$ and $H_R$ are equivalently obtained by considering the restrictions of $G^2$ to its two
bipartition classes.

\begin{definition}[Polymers]
\label{def:polymers}
An $L$-polymer is a set $\gamma\subseteq L$ such that
\begin{enumerate}[label=(\roman*)]
\item $\gamma$ is connected in $H_L$
\item $|\gamma|\leq r=n/\Delta$.
\end{enumerate}
Two $L$-polymers are compatible if their graph distance in $G$ is
greater than two. Equivalently, their distance in the host graph $H_L$
is greater than one, so this is exactly the graph-based compatibility
relation used in Proposition \ref{prop:JKP}. The weight of a polymer is
\begin{equation}
w_\gamma
:=
\frac{\lambda^{|\gamma|}}
{(1+\lambda)^{|N(\gamma)|}}.
\label{eq:polymer-weight}
\end{equation}
The $R$-polymer model is defined symmetrically. Their partition
functions are denoted by $\Xi_L$ and $\Xi_R$.
\end{definition}

We now show that the partition function can be approximated by the polymers. Call a set $A\subseteq L$ admissible if every connected component of $A$ in
$H_L$ has size at most $r$. Define admissible subsets of $R$ symmetrically.
Distinct $H_L$-components have disjoint neighborhoods in $R$, and the
components of an admissible set form a compatible polymer configuration.
Consequently,
\begin{align}
(1+\lambda)^n\Xi_L
&=
\sum_{\substack{I\in\I(G)\\I\cap L\text{ is admissible}}}
\lambda^{|I|},
\label{eq:L-phase-identity}\\
(1+\lambda)^n\Xi_R
&=
\sum_{\substack{I\in\I(G)\\I\cap R\text{ is admissible}}}
\lambda^{|I|}.
\label{eq:R-phase-identity}
\end{align}

Note that these two polymer models together all independent sets besides those that have both sides not admissible. They also double count independent sets with both sides admissible. Accordingly, define
\begin{align}
W_{\cap}
&:=
\sum_{\substack{I\in\I(G)\\
I\cap L\text{ and }I\cap R\text{ are admissible}}}
\lambda^{|I|},
\label{eq:def-Wcap}\\
W_{\varnothing}
&:=
\sum_{\substack{I\in\I(G)\\
I\cap L\text{ is not admissible and }I\cap R\text{ is not admissible}}}
\lambda^{|I|}.
\label{eq:def-Wempty}
\end{align}

We will now show that the sum $(1+\lambda)^n\Xi_L + (1+\lambda)^n\Xi_R$  is a good approximation to the hard-core partition function $Z_G(\lambda)$. In particular, the error, bounded by  $\frac{W_{\cap}+W_{\varnothing}}{Z_G(\lambda)}$, is exponentially small.

\begin{lemma} \label{lem:uniform-phase-error}
There are absolute constants $C_8,c_8>0$ such that, if $n\geq\Delta^2$
and
\begin{equation}
\lambda
\geq
\exp\!\left[
C_8\left(
\frac1\Delta+
\frac{\sigma^2}{\Delta^2-\sigma^2}
\right)\log(\e\Delta)
\right]-1,
\label{eq:uniform-phase-lambda}
\end{equation}
then
\begin{equation}
\frac{W_{\cap}+W_{\varnothing}}{Z_G(\lambda)}
\leq
(2n+2)
\exp\!\left[-c_8\frac{n\log(\e\Delta)}{\Delta^2}\right].
\label{eq:uniform-total-phase-error}
\end{equation}
\end{lemma}

\begin{proof}
Write
\[
\ell:=\log(1+\lambda),
\qquad
x:=
\frac{\sigma^2}{\Delta^2}
+
\frac{1-\sigma^2/\Delta^2}{\Delta}.
\]
We will begin by proving three short claims regarding $W_{\varnothing}$ and $W_{\cap}$

\begin{claim}
\label{clm:missing-phase-weight}
\begin{equation}
\frac{W_{\varnothing}}{(1+\lambda)^n}
\leq
\exp\!\left[
-c_0\frac{\Delta-\sigma}{\Delta^2}\,
n\ell
\right].
\label{eq:Wempty-bound}
\end{equation}
\end{claim}

\begin{proof}
If $I\cap L$ is not admissible, it contains an $H_L$-component of size
greater than $r=n/\Delta$, and hence $|I\cap L|>r$. The same statement
holds on $R$. Therefore, the configurations counted by
$W_{\varnothing}$ form a subset of $\mathcal B$ from
\eqref{eq:def-bad-event}. The claim follows from
Lemma \ref{prop:spectral-LRO}.
\end{proof}

\begin{claim}
\label{lem:overlap-size}
If $I\in\I(G)$ and both $A:=I\cap L$ and $D:=I\cap R$ are admissible, then
\begin{equation}
|I|
\leq \frac{2xn}{1+x}.
\label{eq:overlap-size}
\end{equation}
\end{claim}

\begin{proof}
Distinct connected components of $A$ in $H_L$ have disjoint
neighborhoods. Denote these components to be $A_1, A_2, \cdots A_k$ , and by definition each of these components has size at most $r=n/\Delta$. Applying Lemma \ref{lem:tanner} and summing gives

\begin{align}
|N(A)|
&=
\sum_{i=1}^k |N(A_i)| \nonumber \\
&\geq
\sum_{i=1}^k
\frac{\Delta^2}
{\sigma^2+(\Delta^2-\sigma^2)|A_i|/n}
\,|A_i| \nonumber \\
&\geq
\sum_{i=1}^k
\frac{\Delta^2}
{\sigma^2+(\Delta^2-\sigma^2)/\Delta}
\,|A_i| \nonumber \\
&=
\frac{\Delta^2}
{\sigma^2+(\Delta^2-\sigma^2)/\Delta}
\,|A| \nonumber \\
&=
\frac{|A|}{x},
\end{align}

By symmetry, the same bound holds with $D$ in place of $A$. Also, since $I$ is independent, we have that $N(A)\subseteq R\setminus D$ and $N(D)\subseteq L\setminus A.$ This gives that
\[
\frac{\Delta^2}{\sigma^2+(\Delta^2-\sigma^2)/\Delta}|A|+|D|\leq n,
\qquad
\frac{\Delta^2}{\sigma^2+(\Delta^2-\sigma^2)/\Delta}|D|+|A|\leq n.
\]
Adding the above two inequalities gives that
\begin{align}
\left(
\frac{\Delta^2}
{\sigma^2+(\Delta^2-\sigma^2)/\Delta}
+1
\right)(|A|+|D|)
&\leq 2n.
\end{align}
Rearranging gives the result
\begin{align}
|I|
=|A|+|D|
&\leq
\frac{2n}{
\displaystyle
\frac{\Delta^2}
{\sigma^2+(\Delta^2-\sigma^2)/\Delta}
+1
} \nonumber\\
&=
\frac{
2n\bigl(\sigma^2+(\Delta^2-\sigma^2)/\Delta\bigr)
}{
\Delta^2+\sigma^2+(\Delta^2-\sigma^2)/\Delta
} \nonumber\\
&=
\frac{2xn}{1+x}.
\end{align}
\end{proof}

\begin{claim}
\label{clm:overlap-weight}
There are absolute constants $C_7,c_7>0$ such that, if
\begin{equation}
\ell\geq C_7\left(
\frac1\Delta+
\frac{\sigma^2}{\Delta^2-\sigma^2}
\right)\log(\e\Delta),
\label{eq:overlap-ell-condition}
\end{equation}
then
\begin{equation}
\frac{W_{\cap}}{(1+\lambda)^n}
\leq
(2n+1)\exp\!\left[
-2c_7xn\log(\e\Delta)
\right].
\label{eq:Wcap-bound}
\end{equation}
\end{claim}

\begin{proof}
Set $\beta:=x/(1+x)\in(0,1/2)$, and let
$H(y):=-y\log y-(1-y)\log(1-y)$.
By Claim~\ref{lem:overlap-size}, every independent set counted by
$W_{\cap}$ has size at most $2\beta n$. Therefore,
\begin{align}
\frac{W_{\cap}}{(1+\lambda)^n}
&\leq
\frac{1}{(1+\lambda)^n}
\sum_{k\leq 2\beta n}\binom{2n}{k}\lambda^k
\label{eq:overlap-binomial}\\
&\leq
(2n+1)
\exp\left\{
2nH(\beta)+2\beta n(\log\lambda)_+-n\ell
\right\}
\end{align}

where the second inequality follows from the entropy bound $
\binom{2n}{k}
\leq
\exp\left\{2nH\left(\frac{k}{2n}\right)\right\}
\leq
\exp\{2nH(\beta)\}$ and the inequality $
\lambda^k\leq\exp\{2\beta n(\log\lambda)_+\}$. We then bound as follows

\begin{align}
\frac{W_{\cap}}{(1+\lambda)^n}
&\leq
(2n+1)
\exp\left\{
2n\left(
H(\beta)-\left(\frac12-\beta\right)\ell
\right)
\right\}
\nonumber\\
&\leq
(2n+1)
\exp\left\{
2n\left(
x\log(2\e\Delta)-\frac{1-x}{4}\ell
\right)
\right\}
\nonumber\\
&\leq
(2n+1)
\exp\left\{
2n\left(
2x
-
\frac{C_7}{4}(1-x)
\left(
\frac1\Delta+
\frac{\sigma^2}{\Delta^2-\sigma^2}
\right)
\right)\log(\e\Delta)
\right\}
\nonumber\\
&=
(2n+1)
\exp\left\{
2n\left(
2-\frac{C_7}{4}\frac{\Delta-1}{\Delta}
\right)
x\log(\e\Delta)
\right\}
\nonumber\\
&\leq
(2n+1)
\exp\left\{
2n\left(2-\frac{C_7}{8}\right)
x\log(\e\Delta)
\right\}
\nonumber\\
&\leq
(2n+1)
\exp\left\{
-2c_7xn\log(\e\Delta)
\right\}.
\notag
\end{align}
Here, the first inequality uses $(\log\lambda)_+\leq\ell$ and the second inequality follows from the facts that
\[
H(\beta)\leq\beta\log\frac{\e}{\beta},
\qquad
\beta\leq x,
\qquad
\frac1\beta\leq2\Delta,
\qquad
\frac12-\beta\geq\frac{1-x}{4}.
\]
The third inequality follows from
\eqref{eq:overlap-ell-condition} and
$\log(2\e\Delta)\leq2\log(\e\Delta)$. The equality uses
\[
(1-x)
\left(
\frac1\Delta+
\frac{\sigma^2}{\Delta^2-\sigma^2}
\right)
=
\frac{\Delta-1}{\Delta}x,
\]
and the following inequality uses
$(\Delta-1)/\Delta\geq1/2$. Finally, choosing $C_7$ sufficiently large and then choosing an
absolute $c_7>0$ such that
$2-\frac{C_7}{8}\leq-c_7$
gives the final inequality.
\end{proof}

With the above claims in hand, we are now able to prove the Lemma. We compute as follows:

\begin{align}
\frac{W_{\cap}+W_{\varnothing}}{Z_G(\lambda)}
&\leq
\frac{W_{\cap}+W_{\varnothing}}{(1+\lambda)^n}
\nonumber\\
&=
\frac{W_{\cap}}{(1+\lambda)^n}
+
\frac{W_{\varnothing}}{(1+\lambda)^n}
\nonumber\\
&\leq
(2n+1)
\exp\!\left[
-2c_7xn\log(\e\Delta)
\right]
+
\exp\!\left[
-c_0\frac{\Delta-\sigma}{\Delta^2}n\ell
\right]
\nonumber\\
&\leq
(2n+1)
\exp\!\left[
-c\frac{n\log(\e\Delta)}{\Delta}
\right]
+
\exp\!\left[
-c\frac{n\log(\e\Delta)}{\Delta^2}
\right]
\nonumber\\
&\leq
(2n+1)
\exp\!\left[
-c\frac{n\log(\e\Delta)}{\Delta^2}
\right]
+
\exp\!\left[
-c\frac{n\log(\e\Delta)}{\Delta^2}
\right]
\nonumber\\
&=
(2n+2)
\exp\!\left[
-c\frac{n\log(\e\Delta)}{\Delta^2}
\right].
\end{align}
Here, the first inequality uses
$Z_G(\lambda)\geq(1+\lambda)^n$. The second inequality follows from
Claims~\ref{clm:overlap-weight} and~\ref{clm:missing-phase-weight}. The third inequality, the overlap term uses $x\geq1/\Delta$, while
the missing-phase term uses the fugacity assumption \eqref{eq:uniform-phase-lambda}. The final inequality uses
$\Delta^{-1}\geq\Delta^{-2}$, after decreasing the absolute constant
$c>0$ if necessary.
\end{proof}

Now, as we have control over the errors from $W_{\varnothing}$ and $W_{\cap}$, we are finally able to prove that the polymer model is a good approximation for the partition function. For notation, let $\widetilde Z_G(\lambda)$ denote the approximation output by combining the two polymer models. That is,

\begin{equation}
\widetilde Z_G(\lambda)
:=
(1+\lambda)^n(\Xi_L+\Xi_R).
\end{equation}

By inclusion--exclusion,
\begin{equation}
\widetilde Z_G(\lambda)-Z_G(\lambda)
=
W_{\cap}-W_{\varnothing}.
\label{eq:phase-error-identity}
\end{equation}

Additionally, define
\begin{equation}
c(I)
:=
\bbone\{I\cap L\text{ is admissible}\}
+
\bbone\{I\cap R\text{ is admissible}\}
\in\{0,1,2\},
\label{eq:def-multiplicity}
\end{equation}
that is the number of times an independent set $I$ is counted in
$\widetilde Z_G(\lambda)$. The Gibbs measure corresponding to the
polymer approximation is
\begin{equation}
\widetilde\mu(I)
:=
\frac{c(I)\lambda^{|I|}}{\widetilde Z_G(\lambda)}.
\label{eq:def-two-phase-measure}
\end{equation}

\begin{lemma}
\label{lem:TV-comparison}
Suppose that the fugacity condition
\eqref{eq:uniform-phase-lambda} holds. Then
\begin{equation}
\TV(\widetilde\mu,\mu_{G,\lambda})
\leq
\frac{W_{\cap}+W_{\varnothing}}{Z_G(\lambda)}
\leq
(2n+2)
\exp\!\left[
-c_8\frac{n\log(\e\Delta)}{\Delta^2}
\right].
\label{eq:TV-phase-bound}
\end{equation}
\end{lemma}

\begin{proof}

Since $c(I)-1$ equals $1$ when both sides of $I$ are admissible,
$-1$ when neither side is admissible, and $0$ otherwise, we have that
\[
\sum_{I\in\I(G)}
|c(I)-1|\lambda^{|I|}
=
W_{\cap}+W_{\varnothing}.
\]
which gives
\[
\left|
\widetilde Z_G(\lambda)-Z_G(\lambda)
\right|
\leq
W_{\cap}+W_{\varnothing}.
\]
Therefore,
\begin{align}
2\TV(\widetilde\mu,\mu_{G,\lambda})
&=
\sum_{I\in\I(G)}
\lambda^{|I|}
\left|
\frac{c(I)}{\widetilde Z_G(\lambda)}
-
\frac{1}{Z_G(\lambda)}
\right|
\nonumber\\
&\leq
\sum_{I\in\I(G)}
c(I)\lambda^{|I|}
\left|
\frac{1}{\widetilde Z_G(\lambda)}
-
\frac{1}{Z_G(\lambda)}
\right|
+
\frac{1}{Z_G(\lambda)}
\sum_{I\in\I(G)}
|c(I)-1|\lambda^{|I|}
\nonumber\\
&=
\frac{
|\widetilde Z_G(\lambda)-Z_G(\lambda)|
}{
Z_G(\lambda)
}
+
\frac{W_{\cap}+W_{\varnothing}}{Z_G(\lambda)}
\nonumber\\
&\leq
2\frac{W_{\cap}+W_{\varnothing}}{Z_G(\lambda)}.
\end{align}
Dividing by two and applying Lemma~\ref{lem:uniform-phase-error} gives the result.
\end{proof}

\subsection{Verifying the Koteck\'y--Preiss Criterion}

Having established that the two polymer models are a good approximation to the hard-core model, we must now establish that their partition functions are efficiently computable using Proposition~\ref{prop:JKP}. In particular, we must verify the Koteck\'y--Preiss criterion.

\begin{lemma}
\label{lem:polymer-convergence}
There is an absolute constant $C_9>0$ such that, if
\begin{equation}
\lambda
\geq
\exp\!\left[
C_9\left(
\frac1\Delta+
\frac{\sigma^2}{\Delta^2-\sigma^2}
\right)\log(\e\Delta)
\right]-1,
\label{eq:polymer-lambda-condition}
\end{equation}
then both polymer models satisfy the Koteck\'y--Preiss condition
\eqref{eq:KP-criterion} with the decay function
\begin{equation}
\mathfrak g(\gamma):=\log(\e\Delta)|\gamma|.
\label{eq:def-decay}
\end{equation}
Consequently, each of $\Xi_L$ and $\Xi_R$ admits an FPTAS and a
polynomial-time approximate sampler.
\end{lemma}

\begin{proof}
Recall the notation
\[
\ell:=\log(1+\lambda),
\qquad
\frac{1}{x}:=
\frac{\Delta^2}
{\sigma^2+(\Delta^2-\sigma^2)/\Delta}.
\]

By \eqref{eq:square-degree}, each host graph has maximum degree at
most \(\Delta^2\). Thus, by the standard connected-set counting bound (see, for example,
Lemma~14 of~\cite{JenssenKeevashPerkins2020}),
for every host-graph vertex \(v\) and every \(t\geq1\),
\begin{equation}
\#\left\{
\gamma:
v\in\gamma,\ |\gamma|=t
\right\}
\leq
(\e\Delta^2)^t.
\label{eq:polymer-connected-set-count}
\end{equation}
Moreover, every polymer satisfies \(|\gamma|\leq r=n/\Delta\).
Therefore, by Lemma~\ref{lem:tanner},
\begin{align}
|N(\gamma)|
&\geq
\frac{\Delta^2}
{\sigma^2+(\Delta^2-\sigma^2)|\gamma|/n}
|\gamma|
\nonumber\\
&\geq
\frac{\Delta^2}
{\sigma^2+(\Delta^2-\sigma^2)/\Delta}
|\gamma|
=
\frac{|\gamma|}{x}.
\label{eq:polymer-neighborhood-bound}
\end{align}
Thus, by the definition of the polymer weight,
\begin{align}
w_\gamma
&=
\frac{\lambda^{|\gamma|}}
{(1+\lambda)^{|N(\gamma)|}}
\nonumber\\
&\leq
\exp\!\left[
|\gamma|
\left(
\log\lambda-\frac{\ell}{x}
\right)
\right].
\label{eq:polymer-weight-exponential}
\end{align}

We now verify the Koteck\'y--Preiss criterion. Fix a polymer
\(\gamma\). Every polymer incompatible with \(\gamma\) contains a
vertex in the same bipartition class at graph distance at most two
from \(\gamma\). The number of such vertices is at most
\[
\bigl(1+\Delta(\Delta-1)\bigr)|\gamma|
\leq
\Delta^2|\gamma|.
\]
Consequently, using \eqref{eq:polymer-connected-set-count},
\eqref{eq:polymer-weight-exponential}, and
\eqref{eq:def-decay}, we obtain
\begin{align}
&\sum_{\gamma':\,\gamma'\not\sim\gamma}
w_{\gamma'}
\exp\!\left(
|\gamma'|+\mathfrak g(\gamma')
\right)
\nonumber\\
&\qquad\leq
\Delta^2|\gamma|
\max_v
\sum_{\gamma'\ni v}
w_{\gamma'}
\exp\!\left(
|\gamma'|+\mathfrak g(\gamma')
\right)
\nonumber\\
&\qquad\leq
\Delta^2|\gamma|
\sum_{t\geq1}
\exp\!\left[
t\left(
\log(\e\Delta^2)
+\log\lambda-\frac{\ell}{x}
+1+\log(\e\Delta)
\right)
\right]
\nonumber\\
&\qquad\leq
\Delta^2|\gamma|
\sum_{t\geq1}
\exp\!\left[
t\left(
3\log(\e\Delta)
-
\left(\frac1x-1\right)\ell
\right)
\right]
\nonumber\\
&\qquad\leq
\Delta^2|\gamma|
\sum_{t\geq1}
\exp\!\left[-3t\log(\e\Delta)\right]
\nonumber\\
&\qquad=
\frac{\Delta^2}{(\e\Delta)^3-1}|\gamma|
\leq
|\gamma|.
\label{eq:KP-verification}
\end{align}
where the penultimate inequality follows from
\eqref{eq:polymer-lambda-condition}, as $\left(\frac1x-1\right)^{-1}
=
\frac{\Delta}{\Delta-1}
\left(
\frac1\Delta+
\frac{\sigma^2}{\Delta^2-\sigma^2}
\right)
\nonumber\\
\leq
2\left(
\frac1\Delta+
\frac{\sigma^2}{\Delta^2-\sigma^2}
\right)$.
\end{proof}

The counting and sampling algorithms that provide Theorem \ref{thm:cluster-expansion} then follow exactly as in \cite{JenssenKeevashPerkins2020}.

\subsection{Proof of Theorem \ref{thm:all_fug}} \label{subsec:param_matching}

The proof of Theorem \ref{thm:all_fug} follows from parameter matching our moderate and high fugacity results. 

\begin{proof}[Proof of Theorem \ref{thm:all_fug}]
Let $C$ be the absolute constant in
Theorem \ref{thm:cluster-expansion}. Fix $\xi=1/2$ in
Theorem \ref{thm:main}, and choose an absolute constant $c>0$ sufficiently small
that $4Cc^3\leq1/2$. Choose $\Delta_0$ sufficiently large so
that, whenever $\Delta\geq\Delta_0$ and
\eqref{eq:all-fugacity-spectral-range} holds, we have
$\sigma\leq\Delta/2$ and
\[
C\frac{\sigma^2}
{\Delta^2-\sigma^2}\log(\e\Delta)
\leq1.
\]
If $\sigma=0$, the conclusion follows
directly from Theorem \ref{thm:cluster-expansion}. Otherwise, the inequality
$\e^x-1\leq2x$ for $0\leq x\leq1$ gives
\[
\exp\!\left[
C\frac{\sigma^2}{\Delta^2-\sigma^2}\log(\e\Delta)
\right]-1
\leq
4C\frac{\sigma^2\log(\e\Delta)}{\Delta^2}
\leq
\frac{1}{2\sigma}.
\]
Thus the fugacity ranges in Theorems
\ref{thm:main} and \ref{thm:cluster-expansion} overlap: the former applies up
to $1/(2\sigma)$, and the latter applies from a value no larger
than $1/(2\sigma)$. Together they cover every $\lambda>0$.
\end{proof}

\section*{Statement on AI Use}

ChatGPT 5.5 Pro was used to assist with and verify computations. The authors assume responsibility for all content.

\section*{Acknowledgments}

WP supported in part by NSF grant CCF-2309708.

\newcommand{\etalchar}[1]{$^{#1}$}

\end{document}